\documentclass[12pt]{article}
\usepackage[margin=1in]{geometry}
\usepackage{url}

\usepackage[T1]{fontenc}
\usepackage{fouriernc}

\usepackage{amsmath}
\usepackage{amssymb}
\usepackage{amsthm}
\usepackage{mathrsfs}

\usepackage{hyperref}
\usepackage{xcolor} 

\theoremstyle{plain}
\newtheorem{theorem}{Theorem}
\newtheorem{corollary}[theorem]{Corollary}
\newtheorem{lemma}[theorem]{Lemma}
\newtheorem{proposition}[theorem]{Proposition}

\theoremstyle{definition}
\newtheorem{definition}[theorem]{Definition}
\newtheorem{example}[theorem]{Example}

\newtheorem{problem}[theorem]{Problem}

\theoremstyle{remark}
\newtheorem{remark}[theorem]{Remark}

\newcommand\fib{\boldsymbol{f}}
\newcommand\mbonacci{\boldsymbol{w_m}}
\newcommand\kstrictepi{\boldsymbol{s}}

\title{Palindromic Ziv--Lempel and
Crochemore Factorizations of $m$-Bonacci Infinite Words}
\author{
  Marieh Jahannia\\
  School of Mathematics, Statistics and Computer Science, College of Science\\
University of Tehran, Tehran, Iran\\
  \url{mjahannia@ut.ac.ir}\\
  \and
  Morteza Mohammad-noori\\
  School of Mathematics, Statistics and Computer Science, College of Science\\
  University of Tehran, Tehran, Iran\\
  School of Mathematics, Institute for Research in Fundamental
  Sciences (IPM)\\
  Tehran, Iran\\
  \url{mmnoori@ut.ac.ir}\\
  \and
Narad Rampersad\\
Department of Mathematics and Statistics\\
University of Winnipeg, Winnipeg, Canada\\
\url{n.rampersad@uwinnipeg.ca}\\
\and
Manon Stipulanti\\
Department of Mathematics\\
University of Li\`ege, Li\`ege, Belgium\\
\url{m.stipulanti@uliege.be}}
\date{\today}

\begin{document}
\maketitle

\begin{abstract}
We introduce a variation of the Ziv--Lempel and
Crochemore factorizations of words by requiring each factor to be a
palindrome.  We compute these factorizations for the Fibonacci word, and
more generally, for all $m$-bonacci words.
\end{abstract}

\bigskip
\hrule
\bigskip



\bigskip
\hrule
\bigskip

\section{Introduction}
The Ziv--Lempel \cite{LZ76} and Crochemore \cite{Cro83} factorizations are two well-known factorizations of words used in text compression and other text algorithms.  
Here we apply them to infinite words.  Let $|u|$ denote the length of a finite word $u$.
In this paper, we start indexing words at $0$, i.e., if $u$ is a finite word over the alphabet $A$, then we write $u=u_0\cdots u_{|u|-1}$ where $u_i\in A$ for all $0\le i < |u|$.
If $\boldsymbol{w}$ is an infinite word and $u$ is a finite word, we say
there is an \emph{occurrence of $u$ at position $j$ in $\boldsymbol{w}$} if $\boldsymbol{w} = pu\boldsymbol{w}'$
for some word $p$ of length $j$ and some infinite word $\boldsymbol{w}'$.
Given an infinite word $\boldsymbol{w}$, the \emph{Ziv--Lempel} or \emph{$z$-factorization} of $\boldsymbol{w}$ is the factorization
\[ z(\boldsymbol{w}) = (z_1, z_2, z_3, \ldots) \]
where $z_i$ is the shortest prefix of $z_iz_{i+1}z_{i+2}\cdots$ such that there is no occurrence of $z_i$ in $\boldsymbol{w}$
at any position $j < |z_1z_2\cdots z_{i-1}|$. 
The \emph{Crochemore} or  \emph{$c$-factorization} of $\boldsymbol{w}$ is the factorization
\[ c(\boldsymbol{w}) = (c_1, c_2, c_3, \ldots) \]
where $c_i$ is the longest prefix of $c_ic_{i+1}c_{i+2}\cdots$ such that there is an occurrence of $c_i$ in
$\boldsymbol{w}$ at some position $j < |c_1c_2\cdots c_{i-1}|$, or, if this prefix does not exist, the factor
$c_i$ is just a single letter.

For instance, if $\fib$ is the Fibonacci word, we have
\[
z(\fib) = (0, 1, 00, 101, 00100, 10100101, \ldots)
\]
and
\[
c(\fib) = (0, 1, 0, 010, 10010, 01010010, \ldots).
\]

Note that if $\boldsymbol{w}$ is ultimately periodic the $z$-factorization is not well-defined, since
eventually there will be no factors that do not occur previously in $\boldsymbol{w}$.  Similarly, if $\boldsymbol{w}$ is ultimately periodic the definition of the $c$-factorization will result in some factor $c_i$ being an
infinite word.  We are not interested in ultimately periodic words in this paper and will therefore
ignore this possibility and assume that any infinite word considered in this paper is aperiodic.

In the context of combinatorics on words, these factorizations have
been computed for certain important families of words.  
Berstel and Savelli \cite{BS06} computed the $c$-factorizations of all standard Sturmian words.  
They also observed that the $z$-factorization of the Fibonacci word coincides with the \emph{singular factorization} of the Fibonacci word introduced by Wen and Wen \cite{Wen94}. 
Fici \cite{Fici15} has given an excellent survey of these and other factorizations of the Fibonacci word. 
Ghareghani, Mohammad--noori, and Sharifani \cite{GMS11} determined the $z$- and $c$-factorizations of standard episturmian words.  
Constantinescu and Ilie \cite{CI07} used the $z$-factorization to define the \emph{Lempel--Ziv complexity} of an infinite word.

We introduce the \emph{palindromic $z$-factorization} $pz(\boldsymbol{w})$ and \emph{palindromic $c$-factorization} $pc(\boldsymbol{w})$ by requiring that each of the factors in the previous definitions be palindromes.  
That is, the \emph{palindromic $z$-factorization} of $\boldsymbol{w}$ is the factorization
\[ pz(\boldsymbol{w}) = (z_1, z_2, z_3, \ldots) \]
where $z_i$ is the shortest \emph{palindromic} prefix of $z_iz_{i+1}z_{i+2}\cdots$ such that there is no occurrence of $z_i$ in
$\boldsymbol{w}$ at any position $j < |z_1z_2\cdots z_{i-1}|$.  The palindromic $z$-factorization may
not exist for certain infinite words $\boldsymbol{w}$.  For instance, if $\boldsymbol{w}$ only contains palindromes
of bounded length, then the palindromic $z$-factorization will not exist.  This type of factorization is therefore
only interesting when applied to infinite words with arbitrarily long palindromic factors.
The \emph{palindromic $c$-factorization} of $\boldsymbol{w}$ is the factorization
\[ pc(\boldsymbol{w}) = (c_1, c_2, c_3, \ldots) \]
where $c_i$ is the longest \emph{palindromic} prefix of $c_ic_{i+1}c_{i+2}\cdots$ such that there is an occurrence of $c_i$ in
$\boldsymbol{w}$ at some position $j < |c_1c_2\cdots c_{i-1}|$, or, if this prefix does not exist, the factor
$c_i$ is just a single letter.

For instance,
if $\fib$ is the Fibonacci word, we have
\[
pz(\fib) = (0, 1, 00, 101, 00100, 10100101, \ldots)
\]
and
\[
pc(\fib) = (0, 1, 0, 010, 1001, 0010100, \ldots).
\]
It turns out that $pz(\fib)$ and $z(\fib)$ are the same, and in fact are equal to the \emph{singular factorization} of $\fib$ (which we define later).  
However, the factorizations $pc(\fib)$ and $c(\fib)$ are not the same.  
We show that the factors of $pc(\fib)$ can also be written in
terms of the singular words and the factorization $pc(\fib)$ (except
for the first few factors) coincides with a nice factorization of
$\fib$ that appears in \cite{Fici15}.

We believe that it could be of interest to compare the ordinary $z$-
and $c$-factorizations of certain infinite words with their palindromic $z$- and $c$-factorizations, in the same way that one can compare the ordinary complexity function of an infinite word with its palindromic complexity function (see \cite{ABCD03}).  

The main results of this paper give a description of the palindromic
$z$- and $c$-factorizations of the Fibonacci word and, more generally, the $m$-bonacci word for $m \geq 2$.

\section{Basics from combinatorics on words}\label{sec:basics-cow}

Let $A$ be a finite \emph{alphabet}, i.e., a finite set made of \emph{letters}. 
A \emph{(finite) word} $w$ over $A$ is a finite sequence of letters belonging to $A$. 
If $w=w_0 w_1 \cdots w_n \in A^*$ with $n\ge 0$ and $w_i \in A$ for all $i\in\{0,\ldots,n\}$, then the \emph{length $|w|$} of $w$ is $n+1$, i.e., it is the number of letters that $w$ contains.
We let $\varepsilon$ denote the empty word.  
This special word is the neutral element for concatenation of words, and its length is set to be $0$. 
The set of all finite words over $A$ is denoted by $A^*$, and we let $A^+=A^*\setminus\{\varepsilon\}$ denote the set of non-empty finite words over $A$. 
An \emph{infinite word} $\boldsymbol{w}$ over $A$ is any infinite sequence over $A$. The set of all infinite words over $A$ is denoted by $A^\omega$. 
Note that in this paper infinite words are written in bold.

A finite word $w\in A^*$ is a \emph{prefix} (resp., \emph{suffix}) of another finite word $z\in A^*$ if there exists $u\in A^*$ such that $z=wu$ (resp., $z=uw$).  
The word $w\in A^*$ is said to be a \emph{factor} of $z\in A^*$ if there exist $u,v\in A^*$ such that $z=uwv$.  
If $z=xy$ is a finite word over $A$, we write $x^{-1}z=y$ and $z y^{-1}= x$.
Observe that if $z=xyt$ with $t,x,y,z \in A^*$, then $(xy)^{-1}z=y^{-1}(x^{-1}z)$ and $z(yt)^{-1}=(zt^{-1})y^{-1}$. 
In particular, for any words $u,v \in A^*$, we have $(uv)^{-1}=v^{-1}u^{-1}$.

In the same way, a finite word $w\in A^*$ is a \emph{prefix} of an infinite word $\boldsymbol{z}\in A^\omega$ if there exist $\boldsymbol{u}\in A^\omega$ such that $z=w\boldsymbol{u}$.  
The word $w\in A^*$ is said to be a \emph{factor} of $z\in A^\omega$ if there exist $u\in A^*$ and $\boldsymbol{v}\in A^\omega$ such that $\boldsymbol{z}=uw\boldsymbol{v}$.  

Let $w=w_0 w_1 \cdots w_n \in A^*$ with $n\ge 0$ and $w_i \in A$ for all $i\in\{0,\ldots,n\}$.
The \emph{mirror image}, or \emph{reversal}, of $w$ is the word $w^R=w_n w_{n-1} \cdots w_0$ over $A$, i.e., the word obtained by reading $w$ from right to left. 
We say that a word $w$ over $A$ is a \emph{palindrome} if $w^{R}=w$. 

A \emph{factorization} of a finite word $w \in A^*$ is a finite sequence $(x_n)_{0 \le n \le m}$ of finite words over $A$ such that
$$
w = \prod_{n = 0}^{m} x_n.
$$
Similarly, a \emph{factorization} of an infinite word $\boldsymbol{w}\in A^\omega$ is a sequence $(x_n)_{n\ge 0}$ of finite words over $A$ such that 
$$
\boldsymbol{w} = \prod_{n\ge 0} x_n.
$$

A \emph{morphism} on $A$ is a map $\sigma : A^* \to A^*$ such that for all $u,v \in A^*$, we have $\sigma(uv)=\sigma(u)\sigma(v)$. 
In order to define a morphism, it suffices to provide the image of letters belonging to $A$. 
A morphism is said to be \emph{prolongable} on a letter $a\in A$ if $\sigma(a)=au$ with $u\in A^+$ and $\sigma$ is non-erasing, i.e., the image of no letter is the empty word.
If $\sigma$ is prolongable on $a$, then $\sigma^n(a)$ is a proper prefix of $\sigma^{n+1}(a)$ for all $n\ge 0$. 
Therefore, the sequence $(\sigma^n(a))_{n\ge 0}$ of finite words defines an infinite word $\boldsymbol{w}$ that is a fixed point of $\sigma$.  

In combinatorics on words, given an alphabet $A$, a set $X\subset A^+$ of non-empty words is a \emph{code} on $A$ if any word $w\in A^*$ has at most one factorization using words of $X$. For more on this topic, see, for instance,~\cite[Chapter 6]{Lothaire02}.  
The following result can be found in~\cite[Chapter 6]{Lothaire02}.

\begin{proposition}\label{pro:general-code}
Let $A, B$ be two finite alphabets, and let $\sigma : A^* \to B^*$ be an injective morphism. If $X\subset A^+$ is a code on $A$, then $\sigma(X)$ is a code on $B$.
\end{proposition}

In the following definition, we introduce two new factorizations of interest. 

\begin{definition}\label{def:pal-z-c-fact}
Let $\boldsymbol{w}$ be an infinite word over $A$.
The \emph{palindromic Ziv--Lempel} or
\emph{palindromic $z$-factorization} of $\boldsymbol{w}$ is the factorization
\[ 
pz(\boldsymbol{w}) = (z_1, z_2, z_3, \ldots) 
\]
where $z_i$ is the shortest \emph{palindromic} prefix of $z_iz_{i+1}z_{i+2}\cdots$ such that there is no occurrence of $z_i$ in
$\boldsymbol{w}$ at any position $j < |z_1z_2\cdots z_{i-1}|$.
The \emph{palindromic Crochemore} or 
\emph{palindromic $c$-factorization} of $\boldsymbol{w}$ is the factorization
\[ 
pc(\boldsymbol{w}) = (c_1, c_2, c_3, \ldots) 
\]
where $c_i$ is the longest \emph{palindromic} prefix of $c_ic_{i+1}c_{i+2}\cdots$ such that there is an occurrence of $c_i$ in
$\boldsymbol{w}$ at some position $j < |c_1c_2\cdots c_{i-1}|$, or, if this prefix does not exist, the factor
$c_i$ is just a single letter.
\end{definition}

\section{The Fibonacci case}\label{sec:Fib}

\subsection{Some known results and preliminaries}
Before establishing the two palindromic factorizations of the Fibonacci word, we gather some definitions and necessary results. 
Some of them are well known and can be found in~\cite{Fici15, Wen94}. 
In the following definition, we follow the lines of~\cite{Fici15}.

\begin{definition}\label{def:Fibonacci}
Let $\fib$ be the (infinite) Fibonacci word, i.e., the fixed point of the morphism $\varphi : 0\mapsto 01, 1 \mapsto 0$, starting with $0$. 
For all $n\ge 0$, define the finite word $h_n = \varphi^{n}(0)$ to be the $n$th iteration of $\varphi$ on $0$. 
The first few words of the sequence $(h_n)_{n\ge 0}$ are $0, 01, 010, 01001$.
It is well known that the Fibonacci word $\fib$ is the limit of $(h_n)_{n\ge 0}$.
Let $(p_n)_{n\ge 3}$ be the sequence of the palindromic prefixes of $\fib$, which are also called \emph{central words}. 
The first few terms of this sequence are $\varepsilon, 0, 010, 010010, \ldots$.
The \emph{singular words} $(\hat{f}_n)_{n\ge 1}$ satisfy $\hat{f}_1=0$, $\hat{f}_2=1$ and, for all $n\ge 1$, $\hat{f}_{2n+1}=0p_{2n+1}0$ and $\hat{f}_{2n+2}=1p_{2n+2}1$. 
The first few singular words are $0,1,00, 101, 00100$.
\end{definition}

The following properties of the singular words can be found in~\cite{Wen94}.

\begin{proposition}\label{pro:properties-of-singular-words}
Let $(F_n)_{n\ge 0}$ be the sequence of Fibonacci numbers with initial conditions $F_0=0$ and $F_1=1$.
\begin{itemize}
\item[(1)] For all $n\ge 1$, $\hat{f}_{n}$ is a palindrome.
\item[(2)] For all $n\ge 1$, $|\hat{f}_{n}|=F_n$.
\item[(3)] For all $n\ge 4$, $\hat{f}_{n}=\hat{f}_{n-2} \hat{f}_{n-3} \hat{f}_{n-2}$.
\item[(4)] For all $n\ge 1$, $\hat{f}_{n}$ is not a factor of $\hat{f}_{n+1}$.
\item[(5)] For all $n\ge 1$, $\hat{f}_{n}$ is not a factor of $\prod_{m=1}^{n-1} \hat{f}_{m}$.
\item[(6)] Let $n\ge 1$ and let $\hat{f}_{n+1}=wa$ where $w \in \{0,1\}^{*}$ and $a\in \{0,1\}$. If $\hat{f}_{n+1}'=w\overline{a}$ with $\overline{a}=1-a$, then $\hat{f}_{n+2} = \hat{f}_{n} \hat{f}_{n+1}'$.
\item[(7)] Let $n\ge 3$ and define $\alpha_n$ to be $0$ if $n$ is odd, or $1$ if $n$ is even. 
Then $\hat{f}_n = \alpha_n \prod_{m=1}^{n-2} \hat{f}_{m}$.
\end{itemize}
\end{proposition}

The following result can be found in~\cite{Fici15}.
Note that the first factorization of the Fibonacci word $\fib$ also appears in~\cite{Wen94}.

\begin{proposition}\label{pro:fact-Fici-Fibonacci}
We have the following two factorizations of the Fibonacci word
\begin{eqnarray}
\fib
&=& \prod_{n\ge 1} \hat{f}_n \label{eq:fact1} \\
&=& 0\cdot 1 \cdot 00 \cdot 101 \cdot 00100 \cdot 10100101 \cdots \nonumber \\
&=& 010 \prod_{n\ge 2} \hat{f}_{n-1} \hat{f}_n \hat{f}_{n-1}  \label{eq:fact2} \\
&=& 010 \cdot (0 \cdot 1 \cdot 0) \cdot (1 \cdot 00 \cdot 1) \cdot (00 \cdot 101 \cdot 00) \cdot (101 \cdot 00100 \cdot 101) \cdots. \nonumber
\end{eqnarray}
Moreover, the Ziv--Lempel factorization of the Fibonacci word is given by the sequence of singular words, i.e.,
$$
z(\fib) = (\hat{f}_1, \hat{f}_2, \hat{f}_3, \ldots) .
$$
\end{proposition}

As a matter of fact, the palindromic $z$-factorization of $\fib$ is easily deduced from the previous result, as shown in the next section. 
However, the palindromic $c$-factorization of $\fib$ cannot be obtained from already known results, and, to that aim, we define a sequence of specific prefixes of $\fib$.

\begin{definition}\label{def:prefix}
For all $n\ge 2$, define 
$$
g_n := 010 \prod_{2 \le m\le n-1} \hat{f}_{m-1} \hat{f}_m \hat{f}_{m-1}.
$$
From~\eqref{eq:fact2}, observe that, for all $n\ge 2$, we have 
$$
\fib = g_n \cdot (\hat{f}_{n-1} \hat{f}_n \hat{f}_{n-1}) \cdot \prod_{ m\ge n+1 } \hat{f}_{m-1} \hat{f}_m \hat{f}_{m-1}.
$$
\end{definition}

Interestingly, the prefix $g_n$ of $\fib$ can be factorized as a particular product of singular words. 

\begin{proposition}\label{pro:prefix-Fib}
For all $n\ge 2$, we have 
\begin{equation}\label{eq:prefix}
g_n = \hat{f}_1 \hat{f}_2 \cdots \hat{f}_{n-1} \hat{f}_n \hat{f}_{n-1} \hat{f}_{n-2}.
\end{equation}
\end{proposition}
\begin{proof}
Proceed by induction on $n\ge 2$. 
The result holds for $n=2$ because $g_2=010= \hat{f}_1 \hat{f}_2 \hat{f}_1$. 
For $n=3$, we get $g_3= 010 \cdot (0\cdot 1 \cdot 0)$ by Definition~\ref{def:prefix} and therefore
$$
g_3 = 0 \cdot 1 \cdot 00 \cdot 1 \cdot 0 = \hat{f}_1 \hat{f}_2 \hat{f}_3 \hat{f}_2 \hat{f}_1,
$$ 
as desired. 
%
%
Assume that $n\ge 3$. 
Now we suppose the result holds up to $n$ and we show it still holds for $n+1$. 
Using Definition~\ref{def:prefix}, we have
$$
g_{n+1} = g_n  ( \hat{f}_{n-1} \hat{f}_{n} \hat{f}_{n-1} ).
$$
By the induction hypothesis, we get
\begin{align*}
g_{n+1} 
&= (\hat{f}_1 \hat{f}_2 \cdots \hat{f}_{n-1} \hat{f}_n \hat{f}_{n-1} \hat{f}_{n-2})  ( \hat{f}_{n-1} \hat{f}_{n} \hat{f}_{n-1} ) \\
&= \hat{f}_1 \hat{f}_2 \cdots \hat{f}_{n-1} \hat{f}_n  ( \hat{f}_{n-1} \hat{f}_{n-2} \hat{f}_{n-1} ) \hat{f}_{n} \hat{f}_{n-1}.
\end{align*}
Since $n+1\ge 4$, Proposition~\ref{pro:properties-of-singular-words} implies that $\hat{f}_{n+1}=\hat{f}_{n-1} \hat{f}_{n-2} \hat{f}_{n-1}$, and we deduce that
$$
g_{n+1} = \hat{f}_1 \hat{f}_2 \cdots \hat{f}_{n-1} \hat{f}_n   \hat{f}_{n+1}  \hat{f}_{n} \hat{f}_{n-1},
$$
which ends the proof. 
\end{proof}

\subsection{The palindromic $z$-factorization of the Fibonacci word}

In this (very) short section, we obtain the palindromic $z$-factorization of the Fibonacci word, which easily follows from already known results. 

\begin{theorem}\label{thm:pal-z-fact-fibonacci}
The palindromic $z$-factorization of the Fibonacci word $\fib$ is 
$$
pz(\fib)=(\hat{f}_1, \hat{f}_2, \hat{f}_3, \ldots).
$$ 
\end{theorem}
\begin{proof}
From Proposition~\ref{pro:fact-Fici-Fibonacci}, $z(\fib)=(\hat{f}_1, \hat{f}_2, \hat{f}_3, \ldots)$. 
Since the factors $\hat{f}_n$ are all palindromes by Proposition~\ref{pro:properties-of-singular-words}, this factorization is also $pz(\fib)$.  
\end{proof}

\subsection{The palindromic $c$-factorization of the Fibonacci word}

In this section, we show that, after the prefix of length $3$, the factorization~\eqref{eq:fact2} coincides with the factorization $pc(\fib)$.  
Note that in this case $pc(\fib)$ and $c(\fib)$ are not the same, since the factors in $c(\fib)$ are not palindromes.

\begin{lemma}\label{lem:common-suffix-prefix-Fib}
For all $n\ge 1$, the only suffix of $\hat{f}_{n}$ that is also a prefix of $\hat{f}_{n+1}$ is the empty word. 
\end{lemma}
\begin{proof}
We proceed by induction on $n\ge 1$.
From Definition~\ref{def:Fibonacci}, the first two singular words are  $\hat{f}_{1}=0$ and $\hat{f}_{2}=1$, so the result can be checked by hand for $n=1$.

Now suppose that $n\ge 2$, and that the only suffix of $\hat{f}_{k}$ that is also a prefix of $\hat{f}_{k+1}$ is the empty word, for all $k\in\{1, \ldots, n-1\}$.
We show that the result still holds for $k=n$. 
Proceed by contradiction and suppose there exists a word $x\in\{0,1,\ldots,m-1\}^{*}$ which is a non-empty suffix of $\hat{f}_{n}$ and a non-empty prefix of $\hat{f}_{n+1}$. We have $1 \le |x| \le |\hat{f}_{n}|$. 
Using Proposition~\ref{pro:properties-of-singular-words}(6), $\hat{f}_{n+1}$ starts and ends with $\hat{f}_{n-1}$. 

If $1 \le |x| \le |\hat{f}_{n-1}|$, then $x$ is a prefix of $\hat{f}_{n-1}$ (recall that $x$ is a prefix of $\hat{f}_{n+1}$).
Consequently, $x^R$ is a non-empty suffix of $\hat{f}_{n-1}$ and a non-empty prefix of $\hat{f}_{n}$. This contradicts the inductive assumption. 

If $|\hat{f}_{n-1}| \le |x| \le |\hat{f}_{n}|$, then $\hat{f}_{n-1}$ is a prefix of $x$ (recall that $x$ is a prefix of $\hat{f}_{n+1}$). 
In particular, $\hat{f}_{n-1}$ is a factor of $x$, and also a factor of $\hat{f}_{n}$ (recall that $x$ is a suffix of $\hat{f}_{n}$).
This contradicts Proposition~\ref{pro:properties-of-singular-words}(4).
\end{proof}

In the following lemma, recall that we start indexing words at $0$.

\begin{lemma}\label{lem:two-occurrences}
Let $n\ge 1$. 
There are exactly two occurrences of the factor $\hat{f}_{n}$ inside the word $g_{n+1}$: one at position $\sum_{m=1}^{n-1} | \hat{f}_{m} | $, the other at position $\sum_{m=1}^{n+1} | \hat{f}_{m} | $.
\end{lemma}
\begin{proof}
If $n=1$, then $g_2 = 0 1 0$ and the factor $\hat{f}_{1}=0$ occurs in $g_2$ at positions $0$ and $2=|\hat{f}_{1}|+|\hat{f}_{2}|$.
If $n=2$, then $g_3 = 0 1 00 1 0 = g_{3,0} \cdots g_{3,5}$ with $g_{3,i}\in\{0,1\}$ for all $1\le i \le 5$. 
There are exactly two occurrences of $\hat{f}_{2}=1$ in $g_3$ starting either at position $1=|\hat{f}_{1}|$ or $4=|\hat{f}_{1}|+|\hat{f}_{2}|+|\hat{f}_{3}|$.

Suppose that $n\ge 3$. 
Using~\eqref{eq:prefix}, let us write $g_{n+1} = p \hat{f}_n \hat{f}_{n+1} \hat{f}_n \hat{f}_{n-1}$ with $p = \hat{f}_1 \hat{f}_2 \cdots \hat{f}_{n-1}$. 
Thanks to this factorization, we immediately see that $\hat{f}_{n}$ occurs at least twice as a factor of $g_{n+1}$: one starting at position $|p|=\sum_{m=1}^{n-1} | \hat{f}_{m} |$, the other beginning at position $|p\hat{f}_n \hat{f}_{n+1}|=\sum_{m=1}^{n+1} | \hat{f}_{m} |$. 
We now show that there are no other occurrences of $\hat{f}_{n}$ as a factor of $g_{n+1}$. There are several cases to consider. 

\textbf{Case 1}. The word $\hat{f}_{n}$ cannot be a factor of $p$, otherwise it contradicts Proposition~\ref{pro:properties-of-singular-words}(5).

\textbf{Case 2}. The word $\hat{f}_{n}$ cannot be a factor of $\hat{f}_{n+1}$, otherwise it contradicts Proposition~\ref{pro:properties-of-singular-words}(4).

\textbf{Case 3}. The word $\hat{f}_{n}$ cannot be a factor of $\hat{f}_{n-1}$ since $|\hat{f}_{n-1}|=F_{n-1}<F_{n}=|\hat{f}_{n}|$ by Proposition~\ref{pro:properties-of-singular-words}(2) (note that $n-1\ge 2$).

\textbf{Case 4}. Suppose that $\hat{f}_{n}$ is a factor of $p\hat{f}_{n}$, overlapping $p$ and $\hat{f}_{n}$.
Using Proposition~\ref{pro:properties-of-singular-words}(2) ($n-2\ge 1$), we know that
$$
|\hat{f}_{n}| = F_{n} = F_{n-1} + F_{n-2} = |\hat{f}_{n-1}| + |\hat{f}_{n-2}|.  
$$
Consequently, $\hat{f}_{n}$ is a factor of $\hat{f}_{n-2}\hat{f}_{n-1}\hat{f}_{n}$.
If $\hat{f}_{n}$ starts somewhere within $\hat{f}_{n-2}$, or if $\hat{f}_{n}$ starts with the first letter of $\hat{f}_{n-1}$, then $\hat{f}_{n-1}$ is a factor of $\hat{f}_{n}$, which contradicts Proposition~\ref{pro:properties-of-singular-words}(4).
Therefore $\hat{f}_{n}$ must be a factor of $\hat{f}_{n-1}\hat{f}_{n}$, i.e., there exist a non-empty suffix $x$ of $\hat{f}_{n-1}$ and a non-empty prefix $y$ of $\hat{f}_{n}$ such that $\hat{f}_{n} = xy$.
Then $x$ is also a non-empty prefix of $\hat{f}_{n}$, which contradicts Lemma~\ref{lem:common-suffix-prefix-Fib}.

\textbf{Case 5}. Suppose that $\hat{f}_{n}$ is a factor of $\hat{f}_{n}\hat{f}_{n+1}$, overlapping $\hat{f}_{n}$ and $\hat{f}_{n+1}$.
This case is similar to the fourth case above. 
Indeed, observe that, since $n+1\ge 4$, Proposition~\ref{pro:properties-of-singular-words}(3) gives
$$
\hat{f}_{n+1} = \hat{f}_{n-1} \hat{f}_{n-2} \hat{f}_{n-1}.
$$
Using Proposition~\ref{pro:properties-of-singular-words} again, we know that $|\hat{f}_{n}| = F_{n} = F_{n-1} + F_{n-2} = |\hat{f}_{n-1}| + |\hat{f}_{n-2}|$. Consequently, $\hat{f}_{n}$ is a factor of $\hat{f}_{n}\hat{f}_{n-1}\hat{f}_{n-2}$, so $(\hat{f}_{n})^R=\hat{f}_{n}$ is a factor of $(\hat{f}_{n}\hat{f}_{n-1}\hat{f}_{n-2})^R=\hat{f}_{n-2}\hat{f}_{n-1}\hat{f}_{n}$, which is impossible due to the fourth case.

\textbf{Case 6}. Suppose that $\hat{f}_{n}$ is a factor of $\hat{f}_{n+1}\hat{f}_{n}$, overlapping $\hat{f}_{n+1}$ and $\hat{f}_{n}$.
In this case, $(\hat{f}_{n})^{R}=\hat{f}_{n}$ is a factor of $(\hat{f}_{n+1}\hat{f}_{n})^{R}=\hat{f}_{n}\hat{f}_{n+1}$ since the singular words are palindromes. 
As in the fifth case, we raise a contradiction.

\textbf{Case 7}. Suppose that $\hat{f}_{n}$ is a factor of $\hat{f}_{n}\hat{f}_{n-1}$, overlapping $\hat{f}_{n}$ and $\hat{f}_{n-1}$.
In this case, $(\hat{f}_{n})^{R}=\hat{f}_{n}$ is a factor of $(\hat{f}_{n}\hat{f}_{n-1})^{R}=\hat{f}_{n-1}\hat{f}_{n}$ since the singular words are palindromes. 
As in the fourth case, we reach a contradiction.
\end{proof}

We prove a technical result before getting the palindromic $c$-factorization of $\fib$. 


\begin{proposition}\label{pro:nothing-to-add}
Let $n\ge 2$. 
Let $w$ be a non-empty common finite prefix of the infinite words 
$$
 \hat{f}_{n-1}   \cdot ( \hat{f}_{n} \hat{f}_{n+1} \hat{f}_{n}  ) \cdot ( \hat{f}_{n+1} \hat{f}_{n+2} \hat{f}_{n+1}  ) \cdots
$$
and
$$
 ( \hat{f}_{n+1} \hat{f}_{n+2} \hat{f}_{n+1}  ) \cdot  ( \hat{f}_{n+2} \hat{f}_{n+3} \hat{f}_{n+2}  ) \cdots.
$$
Then $\hat{f}_{n} \hat{f}_{n+1} \hat{f}_{n}w$ is not a palindrome.
\end{proposition}
\begin{proof}
Let us define 
$$
u_{n+1} := \hat{f}_{n} \hat{f}_{n+1} \hat{f}_{n}w
$$
where $w$ is taken as in the statement.
Using Proposition~\ref{pro:properties-of-singular-words}, since $\hat{f}_{n+1} = \hat{f}_{n-1}\hat{f}_{n}'$ and $|\hat{f}_{n-1}| + |\hat{f}_{n}| = |\hat{f}_{n+1}|$, we know that $0 < |w| <  |\hat{f}_{n+1}|$.
Now proceed by contradiction and suppose that $u_{n+1}$ is a palindrome.
Then we have 
\begin{align}\label{eq:nothing-to-add}
\hat{f}_{n} \hat{f}_{n+1} \hat{f}_{n}w
= u_{n+1}
= u_{n+1}^R
= w^R \hat{f}_{n} \hat{f}_{n+1} \hat{f}_{n}.
\end{align}
The bounds on the length of $w$ lead to an overlap between the occurrence of $\hat{f}_{n+1}$ at position $|\hat{f}_{n}|$ (in the leftmost word in~\eqref{eq:nothing-to-add}), and the occurrence $\hat{f}_{n}$ at position $|w^R|=|w|$ (in the rightmost word in~\eqref{eq:nothing-to-add}). 
This is impossible due to either Proposition~\ref{pro:properties-of-singular-words}(4), or Lemma~\ref{lem:common-suffix-prefix-Fib}.
\end{proof}

\begin{theorem}\label{thm:pal-c-fact-fibonacci}
Let $pc(\fib)=(c_{-1}, c_0, c_1, c_2, \ldots)$ denote the palindromic $c$-factorization of the Fibonacci word $\fib$. 
Then, we have $c_{-1}=0$, $c_0=1$, $c_1=0$ and, for all 
$n\ge 2$, 
$$
c_n = \hat{f}_{n-1}  \hat{f}_{n} \hat{f}_{n-1}.
$$
\end{theorem}
\begin{proof}
By definition of the palindromic $c$-factorization of the Fibonacci word $\fib$, we clearly have $c_{-1}=0$, $c_0=1$ and $c_1=0$.
For the second part of the result, proceed by induction on $n\ge 2$. 
Suppose $n=2$. 
Let us find the factor $c_2$ of the palindromic $c$-factorization of $\fib$.
We have 
$$
\fib = 0 \cdot 1 \cdot 0 \cdot 0 101 00100 10100101 \cdots
$$
and the longest palindrome starting with $0$ and occurring before is 
$$
c_2 = 010 =  0 \cdot 1 \cdot 0  = \hat{f}_{1}  \hat{f}_{2} \hat{f}_{1},
$$
as expected. 

For the induction step, suppose $n\ge 2$ and assume that, for all $2\le m \le n$, we have $c_m = \hat{f}_{m-1}  \hat{f}_{m} \hat{f}_{m-1}$.
We show it is still true for $m=n+1$. 
On the one hand, by the induction hypothesis, we have
\begin{eqnarray}
\fib
&=& c_{-1} c_0 c_1 c_2 \cdots c_n c_{n+1} c_{n+2} \cdots \nonumber  \\
&=& 010 \left( \prod_{2 \le m \le n} c_m \right) c_{n+1} c_{n+2} \cdots \nonumber  \\
&=& 010 \left( \prod_{2 \le m \le n} \hat{f}_{m-1}  \hat{f}_{m} \hat{f}_{m-1} \right) c_{n+1} c_{n+2} \cdots \label{eq:pal-fact-ind-step}
\end{eqnarray}
and the goal is to find the next factor of the palindromic $c$-factorization of $\fib$, i.e., the word $c_{n+1}$.
On the other hand, using~\eqref{eq:fact2} first and then~\eqref{eq:prefix} since $n$ is large enough, we get 
\begin{eqnarray}
\fib 
&=& 010 \left( \prod_{2 \le m \le n} \hat{f}_{m-1} \hat{f}_m \hat{f}_{m-1} \right) (\hat{f}_{n} \hat{f}_{n+1} \hat{f}_{n}) (\hat{f}_{n+1} \hat{f}_{n+2} \hat{f}_{n+1}) \cdots \nonumber \\
&=& g_{n+1} (\hat{f}_{n} \hat{f}_{n+1} \hat{f}_{n}) (\hat{f}_{n+1} \hat{f}_{n+2} \hat{f}_{n+1}) \cdots \nonumber \\
&=& (\hat{f}_1 \hat{f}_2 \cdots \hat{f}_{n-1} \hat{f}_n \hat{f}_{n+1} \hat{f}_{n} \hat{f}_{n-1}) (\hat{f}_{n} \hat{f}_{n+1} \hat{f}_{n}) (\hat{f}_{n+1} \hat{f}_{n+2} \hat{f}_{n+1}) \cdots. \label{eq:Fici-fact}
\end{eqnarray}
Using~\eqref{eq:Fici-fact}, it is clear that $|c_{n+1}| \ge | \hat{f}_{n} \hat{f}_{n+1} \hat{f}_{n} |$ since $\hat{f}_{n} \hat{f}_{n+1} \hat{f}_{n}$ is a palindrome occurring before.
Therefore, there exists a word $w \in \{0,1\}^{*}$ such that $c_{n+1} = \hat{f}_{n} \hat{f}_{n+1} \hat{f}_{n} w$.
We claim that $w$ is in fact the empty word and proceed by contradiction. 

By Lemma~\ref{lem:two-occurrences}, we know that there are exactly two occurrences of $\hat{f}_{n}$ in $g_{n+1}$: one starts at  position $\sum_{m=1}^{n-1} | \hat{f}_{m} | $, and the other at  position $\sum_{m=1}^{n+1} | \hat{f}_{m} |$.

\textbf{Case 1}. Let us deal with the occurrence of $\hat{f}_{n}$ in $g_{n+1}$ at position $\sum_{m=1}^{n-1} | \hat{f}_{m} | $.
In this case, $w$ must be a common prefix of the infinite words 
$$\hat{f}_{n-1}   \cdot ( \hat{f}_{n} \hat{f}_{n+1} \hat{f}_{n}  ) \cdot ( \hat{f}_{n+1} \hat{f}_{n+2} \hat{f}_{n+1}  ) \cdots
$$
and
$$
 ( \hat{f}_{n+1} \hat{f}_{n+2} \hat{f}_{n+1}  ) \cdot  ( \hat{f}_{n+2} \hat{f}_{n+3} \hat{f}_{n+2}  ) \cdots.
$$ 
By Proposition~\ref{pro:nothing-to-add}, we know that $\hat{f}_{n} \hat{f}_{n+1} \hat{f}_{n}w$ is not a palindrome if $w$ is non-empty, a contradiction. 

\textbf{Case 2}. Let us consider the occurrence of $\hat{f}_{n}$ in $g_{n+1}$ at position $\sum_{m=1}^{n+1} | \hat{f}_{m} | $.
In this case, $\hat{f}_{n} \hat{f}_{n+1} \hat{f}_{n} w$ must be a common prefix of the infinite words 
$$\hat{f}_{n} \hat{f}_{n-1}   \cdot ( \hat{f}_{n} \hat{f}_{n+1} \hat{f}_{n}  ) \cdot ( \hat{f}_{n+1} \hat{f}_{n+2} \hat{f}_{n+1}  ) \cdots
$$
and
$$
(\hat{f}_{n} \hat{f}_{n+1} \hat{f}_{n}) \cdot ( \hat{f}_{n+1} \hat{f}_{n+2} \hat{f}_{n+1}  ) \cdot  ( \hat{f}_{n+2} \hat{f}_{n+3} \hat{f}_{n+2}  ) \cdots.
$$ 
Using Proposition~\ref{pro:properties-of-singular-words}, we know that
$$
|\hat{f}_{n+1}| = F_{n+1} = F_{n} + F_{n-1} = |\hat{f}_{n}| + |\hat{f}_{n-1}|.  
$$
Consequently, $\hat{f}_{n+1} = \hat{f}_{n-1}\hat{f}_{n}$, which violates Proposition~\ref{pro:properties-of-singular-words} (items (4) or (6)).

As a conclusion, the longest palindrome starting with the first letter of $\hat{f}_{n}$ and occurring before is 
$$
c_{n+1} = \hat{f}_{n} \hat{f}_{n+1} \hat{f}_{n},
$$
as required.
%
%
%
\end{proof}

\section{The $m$-bonacci case}\label{sec:mbonacci}

In this section, we extend the results obtained for the Fibonacci word to any $m$-bonacci word, namely we get the palindromic $z$- and $c$-factorizations of any $m$-bonacci word. 
The strategy is similar to the one adopted in the previous case: we define a particular sequence $(z^{(m)}_n)_{n\ge -1}$ of finite words that we will call \emph{p-singular words}, and we write the palindromic $z$- and $c$-factorizations of any $m$-bonacci word in terms of this sequence. 
In the case $m=2$, the words $(z^{(2)}_n)_{n\ge 0}$ turn out to be the singular words $(\hat{f}_n)_{n\ge 1}$ (see Proposition~\ref{pro:f-n-z-n}).

\subsection{Preliminaries}

\begin{definition}\label{def:mbonacci}
Let $m\ge 2$.
We define the morphism $\phi_m $ on $\{0, 1, \ldots, m-1\}$ by
$$
\phi_m : 0\mapsto 01, 1 \mapsto 02, \cdots , (m-2)\mapsto 0(m-1), (m-1)\mapsto 0.
$$
When $m=2$, then $\phi_m = \varphi$, and we fall into the Fibonacci case above. 

Let $\mbonacci$ be the (infinite) $m$-bonacci word, i.e., the fixed point of the morphism $\phi_m$, starting with $0$. 
For all $n\ge 0$, define $h_n^{(m)} = \phi_m^{n}(0)$ to be the $n$th iteration of $\phi_m$ on $0$. 
It is well known that the $m$-bonacci word $\mbonacci$ is the limit of $(h_n^{(m)})_{n\ge 0}$.
For the sake of simplicity, when the context is clear, we write $h_n$ instead of $h_n^{(m)}$.
\end{definition}

From now on, $m$ is a fixed integer greater than $1$ unless otherwise specified. 

\begin{example}

If $m=2$, then $\boldsymbol{w_2}=\fib$ is the Fibonacci word. 
See also Definition~\ref{def:Fibonacci}.
If $m=3$, then $\phi_3 : 0\mapsto 01, 1 \mapsto 02, 2 \mapsto 0$, and the infinite word $\boldsymbol{w_3}$ is called the Tribonacci word. 
If $m=4$, then $\phi_4 : 0\mapsto 01, 1 \mapsto 02, 2 \mapsto 03, 3 \mapsto 0$, and the infinite word $\boldsymbol{w_4}$ is called the Quadribonacci word. 
In Table~\ref{table:first-few-words-h-n}, the first few words of the sequences $(h_n^{(m)})_{n\ge 0}$ are given for $m\in\{2,3,4\}$.
\begin{table}[h]
$$
\begin{array}{c|l}
m & (h_n^{(m)})_{n\ge 0} \\
\hline
2 & 0, 01, 010, 01001, \ldots \\
3 & 0, 01, 0102, 0102010, 0102010010201, \ldots\\
4 & 0, 01, 0102, 01020103, 010201030102010, \ldots \\
\end{array}
$$
\caption{The first few words of the sequences $(h_n^{(m)})_{n\ge 0}$ for $m\in\{2,3,4\}$.}
\label{table:first-few-words-h-n}
\end{table}
\end{example}

\begin{lemma}\label{lem:code-mbonacci}
The set of non-empty words $\{01, 02, \ldots, 0(m-1),0\}$ is a code on the finite alphabet $\{0,1,\ldots, m-1\}$.
\end{lemma}
\begin{proof}
It directly follows from Proposition~\ref{pro:general-code}, the fact that $\phi_m$ is an injective morphism, and $\{0,1,\ldots,m-1\}$ is a code on $\{0,1,\ldots,m-1\}$ .
\end{proof}

\begin{remark}
Observe that not all words over $\{0,1,\ldots,m-1\}$ have a factorization in terms of blocks of $\{01, 02, \ldots, 0(m-1),0\}$.
For instance, the word $1$ does not have any such factorization. 
However if a word has such a factorization, then it is unique. 
\end{remark}

The following lemma will be useful to prove properties similar to those given in Proposition~\ref{pro:properties-of-singular-words}.

\begin{lemma}\label{lem:strange-technical-lemma-v2}
Let $x,y \in \{0,1,\ldots,m-1\}^*$ be two finite words.
\begin{itemize}
\item[(1)] If $\phi_m(x)0$ is a factor of $\phi_m(y)0$, then $x$ is a factor of $y$. 
\item[(2)] If $\phi_m(x)$ is a factor of $\phi_m(y)$ and $x$ does not end with the letter $m-1$, then $x$ is a factor of $y$. 
\end{itemize}
\end{lemma}
\begin{proof}
If $x$ is the empty word, then both items are true.
Now assume that $x$ is non-empty, so is $y$.
From Lemma~\ref{lem:code-mbonacci}, the set of words $C= \{01,02, \ldots, 0(m-1), 0\}$ is a code, and thus the words $\phi_m(x)$ and $\phi_m(y)$ respectively admit a unique factorization in terms of blocks belonging to $C$. 
There exist positive integers $\ell, k$ and words $x_1,\ldots,x_\ell,y_1,\ldots,y_k \in C$ such that
$\phi_m(x)=x_1 \cdots x_\ell$ and $\phi_m(y)=y_1 \cdots y_k$. 

Let us prove item $(1)$. 
Note that $\phi_m(x(m-1))=\phi_m(x)0=x_1 \cdots x_\ell 0$ and $\phi_m(y(m-1))=\phi_m(y)0=y_1 \cdots y_k 0$. 
By assumption, there exist words $w,t \in\{0,1,\ldots,m-1\}^{*}$ such that $\phi_m(y)0 = w \phi_m(x) 0 t$.
Using the form of the blocks in $C$, there exist $1\le i \le j \le k$ such that $y_i=x_1$ and $y_j=x_\ell$ (since no words in $C$ start with a letter different from $0$ and since $00$ implies that the first letter $0$ is a block in $C$). 
By uniqueness of the factorization, we also have $y_{i+r-1}=x_r$ for all $1 \le r \le \ell$ and $j=i+\ell-1$.
Consequently, $w=y_0 y_1 \cdots y_{i-1}$ and $0t= y_{i+\ell} \cdots y_k 0$.
From the form of the words $w$ and $0t$, we deduce that there exist words $w',t' \in\{0,1,\ldots,m-1\}^{*}$ such that $\phi_m(w')=w$ and $\phi_m(t')=0t$. 
Thus,
$$
\phi_m(y)0
= w \cdot \phi_m(x) \cdot (0 t)
= \phi_m(w')\phi_m(x)\phi_m(t') 
=\phi_m(w'xt').
$$
By injectivity of $\phi_m$, $y=w'xt'$, and $x$ is a factor of $y$, as desired. 

Let us show that item $(2)$ also holds. 
By hypothesis, $\phi_m(x)$ is a factor of $\phi_m(y)$ and the last letter of $x$ is not $m-1$ (i.e., the block $x_\ell$ is of length $2$ and ends with a letter different from $0$). 
By an analogous reasoning, there exists $1\le i \le k$ such that $y_{i+r-1}=x_r$ for all $1 \le r \le \ell$.
Now let $w=y_0 y_1 \cdots y_{i-1}$ and $t= y_{i+\ell} \cdots y_k$.
We have
$$
\phi_m(y)
= (y_0 \cdots y_{i-1}) \cdot (y_i  \cdots y_{i+\ell-1}) \cdot (y_{i+\ell} \cdots y_k)
= w \phi_m(x) t.
$$
As before, there exist $w',t' \in\{0,1,\ldots,m-1\}^{*}$ such that $\phi_m(w')=w$ and $\phi_m(t')=t$. 
By injectivity of $\phi_m$, $x$ is again a factor of $y$. 
\end{proof}

\begin{example}
Let $x=1012$ and $y=1010$ be words in $\{0,1,2\}^*$ ($m=3$ here). 
We see that $\phi_m(x)=0201020$ is a factor of $\phi_m(y)=02010201$ while $x$ is not a factor of $y$.
This is due to the fact that $x$ ends with $m-1=2$.
\end{example}

\subsection{Properties of p-singular words}\label{sec:p-singular}

In~\cite{Fici15,Wen94}, the Fibonacci word $\fib$ is factorized into singular words (see Proposition~\ref{pro:fact-Fici-Fibonacci}).
In~\cite{Mel99}, this notion of singular words is extended to cover the case of characteristic Sturmian words.
In particular, any characteristic Sturmian word $\boldsymbol{c_\alpha}$ has a singular decomposition and those singular words are useful to find the palindromic factors of $\boldsymbol{c_\alpha}$.
Leaving the framework of a two-letter alphabet, it is shown in~\cite{TanWen06} that there are two kinds of singular words in the Tribonacci case ($m=3$), and that the Tribonacci word possesses a decomposition into singular words.
Afterwards, the study of singular words has been extended to include standard episturmian words.
More particularly, a standard episturmian $\kstrictepi$ word over $\{a_1, \ldots, a_k\}$ is \emph{$k$-strict} if every letter $a_i$, $1\le i \le k$, occurs infinitely many times in its directive word. 
In fact, the $k$-strict standard episturmian words are exactly the $k$-letter Arnoux--Rauzy sequences. 
To learn more about the subject, we refer the reader to~\cite{Glen06}. 
In~\cite[Chapter 7]{Glen06}, it is shown that any word $\kstrictepi$ in a class of specific $k$-strict standard episturmian words has several kinds of generalized singular words. 
Roughly, those singular words turn out to be notably useful to study factors of $\kstrictepi$ (e.g., squares, cubes and other powers), and can also be used to factorize $\kstrictepi$ (this particular factorization is referred to as a \emph{partition} in ~\cite[Chapter 7]{Glen06}).

Following the same lead, we define the p-singular words in the general case of the $m$-bonacci word $\mbonacci$. 
Those particular words are useful to obtain the palindromic $z$- and $c$-factorizations of $\mbonacci$. 
In this section, we study some of their properties. 

\begin{definition}\label{def:z-n-rec-mbonacci}
Define the sequence $(z_n^{(m)})_{n\ge -1}$ of finite words over the alphabet $\{0,1,\ldots,m-1\}$ by
$z_{-1}^{(m)}=\varepsilon$, $z_{0}^{(m)}=0$, and
\begin{itemize}
\item[(1)] For all $1 \le n\le m-1$, $z_n^{(m)}= z_{n-2}^{(m)} z_{n-3}^{(m)} \cdots z_1^{(m)} z_0^{(m)} n z_0^{(m)} z_1^{(m)} \cdots z_{n-3}^{(m)} z_{n-2}^{(m)}$;
\item[(2)] For all $n\ge m$, $z_n^{(m)}= z_{n-2}^{(m)} z_{n-3}^{(m)} \cdots z_{n-(m-1)}^{(m)} z_{n-m}^{(m)} z_{n-(m+1)}^{(m)} z_{n-m}^{(m)} z_{n-(m-1)}^{(m)}  \cdots z_{n-3}^{(m)} z_{n-2}^{(m)}$. 
\end{itemize}
Note that for $0\le n\le m-1$ (resp. $n\ge m$), $z_n^{(m)}$ is centered at $n$ (resp., $z_{n-(m+1)}^{(m)}$).
\end{definition}

In the Fibonacci case when $m=2$, we will show in Proposition~\ref{pro:f-n-z-n} that the corresponding words are the singular words $(\hat{f}_n)_{n\ge 1}$. 
For that reason, the sequence $(z_n^{(m)})_{n\ge -1}$ is the sequence
of words called \textit{p-singular words}.  The p-singular words
satisfy a number of identities related to the standard and central
words; for instance, the following result 
gives another way we could have chosen to define the p-singular words;
see \cite{GMS11}, where the ordinary $z$- and $c$-factorizations of
episturmian words are best described in terms of the words $h_n^R$.

\begin{lemma}\label{lem:z-n-mbonacci}
For all $n\ge 0$, we have
$$
z_n^{(m)} =
\begin{cases}
(h_1^{(m)})^R (h_3^{(m)})^R \cdots (h_n^{(m)})^R \left((h_0^{(m)})^R (h_2^{(m)})^R \cdots (h_{n-1}^{(m)})^R\right)^{-1},  & \text{if } n \text{ is odd}; \\
(h_0^{(m)})^R (h_2^{(m)})^R \cdots (h_n^{(m)})^R \left((h_1^{(m)})^R (h_3^{(m)})^R \cdots  (h_{n-1}^{(m)})^R \right)^{-1},  & \text{if } n \text{ is even}.
\end{cases}
$$
\end{lemma}

The lemma can be proved directly from the definitions, but we
are able to give a much more elegant proof after first proving some preliminary
results.  We therefore postpone the proof until after Lemma~\ref{lem:z-n-phi-m}.

\begin{example}
In Table~\ref{table:first-few-words-z-n}, the first few p-singular words are displayed for $m\in\{2,3,4,5\}$.
\begin{table}[h]
$$
\begin{array}{c|l}
m & (z_n^{(m)})_{n\ge -1} \\
\hline
2 & \varepsilon, 0,1,00,101,00100, \ldots \\
3 & \varepsilon,0,1,020, 1001, 020 101 020, 1001 020 1 020 1001, \ldots\\
4 & \varepsilon,0,1,020, 10301, 020 1001 020, 10301 020 1 0 1 020 10301, \ldots \\
5 & \varepsilon,0,1,020, 10301, 020 10 4 01 020, 10301 020 1 0 0 1 020 10301, \ldots
\end{array}
$$
\caption{The first few words of the sequences $(z_n^{(m)})_{n\ge -1}$ for $m\in\{2,3,4,5\}$.}
\label{table:first-few-words-z-n}
\end{table}
\end{example}

In fact, in the context of the Fibonacci word ($m=2$), the word $z^{(2)}_n$ is the $(n+1)$st singular word, as shown below. 
As a consequence, the palindromic $z$- and $c$-factorizations of the Fibonacci word can be rewritten in terms of the sequence $(z^{(2)}_n)_{n\ge 0}$ of p-singular words; see Theorems~\ref{thm:pal-z-fact-fibonacci} and~\ref{thm:pal-c-fact-fibonacci}. 
In the same way, we will show that the palindromic $z$- and $c$-factorizations of any $m$-bonacci word involve the p-singular words.

\begin{proposition}\label{pro:f-n-z-n}
For all $n\ge 0$, we have $\hat{f}_{n+1} =z^{(2)}_n $.
\end{proposition}
\begin{proof}
We proceed by induction on $n\ge 0$. 
The result is true for $n\in\{0,1,2\}$.
Now suppose that $n\ge 3$ and the result holds up to $n-1$. 
We show it is still true for $n$.
Using Proposition~\ref{pro:properties-of-singular-words}, then the induction hypothesis and finally Definition~\ref{def:z-n-rec-mbonacci}, we get the result
\[
\hat{f}_{n+1} 
= \hat{f}_{n-1} \hat{f}_{n-2} \hat{f}_{n-1}
= z^{(2)}_{n-2} z^{(2)}_{n-3} z^{(2)}_{n-2}
= z^{(2)}_n. \qedhere
\]
\end{proof}

Again, for the sake of simplicity, when the context is clear, we write $z_n$ instead of $z_n^{(m)}$.
By induction and Definition~\ref{def:z-n-rec-mbonacci}, it is clear
that the p-singular words are palindromes. 

\begin{proposition}\label{pro:z-n-pal-mbonacci}
For all $n\ge -1$, $z_n$ is a palindrome. 
\end{proposition}
%
Also from Definition~\ref{def:z-n-rec-mbonacci}, we know the prefixes and suffixes of length at most $3$ of the p-singular words.

\begin{proposition}\label{pro:z-n-first-last-letters}
For all $n\ge 0$, $z_n$ starts and ends with the letter $0$ (resp., $1$) if $n$ is even (resp., odd). 
Moreover, for all even $n\ge 2$, $z_n$ starts and ends with $00$ if $m=2$, or $020$ if $m\ge 3$; for all odd $n\ge 2$, $z_n$ starts with $101$ if $m=2$, or $100$ if $m=3$, or $103$ if 
$m\ge 4$, and ends with $101$ if $m=2$, or $001$ if $m=3$, or $301$ if $m\ge 4$.
\end{proposition}
\begin{proof}
For the first part of the result, we proceed by induction on $n\ge 0$.
From Definition~\ref{def:z-n-rec-mbonacci}, $z_0=0$ and $z_1=1$, so the result is true for $n\in\{0,1\}$.
Now suppose that $n\ge 2$, and that the result holds for values less than $n$. 
If $1\le n \le m-1$ (resp., $n\ge m$), then Definition~\ref{def:z-n-rec-mbonacci}(1) (resp., Definition~\ref{def:z-n-rec-mbonacci}(2)) shows that $z_n$ ends and starts with $z_{n-2}$.
Using the induction hypothesis since $n-2 \ge 0$, we know that $z_{n-2}$ starts and ends with $0$ (resp., $1$) if $n-2$ is even (resp., odd). 
Consequently, $z_n$ starts and ends with the letter $0$ (resp., $1$) if $n$ is even (resp., odd). 

The proof of the second part of the statement is obtained in the same manner by first observing that Definition~\ref{def:z-n-rec-mbonacci} (or Table~\ref{table:first-few-words-z-n}) gives $z_2=00$ if $m=2$, or $z_2=020$ if $m\ge 3$, and 
$$
z_3 = 
\begin{cases}
101,  & \text{if } m=2; \\
1001,  & \text{if } m=3; \\
10301,  & \text{if } m\ge 4.
\end{cases}
$$
\end{proof}

In the following two corollaries, resulting from Definition~\ref{def:z-n-rec-mbonacci}, we study the length of p-singular words. 

\begin{corollary}\label{cor:comparison-of-length-mbonacci-1}
We have $|z_0|=1$, and for all $1 \le n\le m-1$, $|z_n| = 2 \sum_{k=0}^{n-2} |z_k| +1$.
In particular, for all $1 \le n\le m-1$, $|z_n| = |z_{n-1}| +  2 |z_{n-2}|$.
\end{corollary}
\begin{proof}
From Definition~\ref{def:z-n-rec-mbonacci}, we have $|z_0|=|0|=1$.
Now let $1 \le n\le m-1$.
From Definition~\ref{def:z-n-rec-mbonacci}(1), we get
$$
|z_n| = 2 \sum_{k=0}^{n-2} |z_k| + 1,
$$
which proves the first part of the statement. 
Let us show the second part of the statement.
The case $n=1$ is easily handled. 
Suppose that $2 \le n\le m-1$.
From the first part of the result with $n-1 \in \{1,\ldots,m-1\}$, we know that
\[
|z_n| = 2 \sum_{k=0}^{n-2} |z_k| + 1 = \left( 2 \sum_{k=0}^{n-3} |z_k| + 1 \right) + 2 |z_{n-2}|  = |z_{n-1}| + 2 |z_{n-2}| . 
 \qedhere
\]
\end{proof}

In the following corollary, when $n$ is big enough, the length of the p-singular word $z_n$ is expressed in terms of the length of the previous $m$ p-singular words $z_{n-1}, \ldots, z_{n-m}$. 
Note that, when $m=2$, then the following result is implied by Propositions~\ref{pro:properties-of-singular-words}(2) and~\ref{pro:f-n-z-n}.
Also observe that, when $m$ is even, the sequence $(|z_n|)_{n\ge 0}$ of positive integers satisfies a $m$-bonacci type recurrence relation. 
However, that is not the case when $m$ is odd.

\begin{corollary}\label{cor:comparison-of-length-mbonacci-2}
If $m$ is even, then, for all $n\ge m-1$, we have 
$$
|z_n| =
|z_{n-1}| + |z_{n-2}| + \cdots + |z_{n-m}|.
$$
If $m$ is odd, then, for all $n\ge m-1$, we have 
$$
|z_n| = |z_{n-1}| + |z_{n-2}| + \cdots + |z_{n-m}| + (-1)^n.
$$
\end{corollary}
\begin{proof}
If $m=2$, the result follows from Propositions~\ref{pro:properties-of-singular-words}(2) and~\ref{pro:f-n-z-n}.
Now suppose that $m\ge 3$, and, as a first case, suppose that $m$ is even. 
Proceed by induction on $n\ge m-1$.
If $n=m-1$, then using Corollary~\ref{cor:comparison-of-length-mbonacci-1} several times, we have 
\begin{align*}
|z_{m-1}| &= |z_{m-2} | + 2 |z_{m-3}|  = |z_{m-2} | + |z_{m-3}| + |z_{m-4} | + 2 |z_{m-5}|  =  |z_{m-2} | + |z_{m-3}|  + \cdots + |z_2| + 2 |z_1| \\
&= |z_{m-2} | + |z_{m-3}|  + \cdots + |z_2| +  |z_1| + |z_0| + |z_{-1}| 
\end{align*}
since $|z_1| = |z_0|$, and $|z_{-1}| =0$.
Now suppose that $n\ge m$ and the result holds for values less than $n$. 
From Definition~\ref{def:z-n-rec-mbonacci}, we obtain
\begin{align*}
|z_n| &=  2 |z_{n-2}| + 2 |z_{n-3}| + \cdots + 2 |z_{n-m}| + |z_{n-(m+1)}| \\
&= (|z_{n-2}| +  |z_{n-3}| + \cdots +  |z_{n-1-(m-1)}| + |z_{n-1-m}| ) + |z_{n-2}| +  |z_{n-3}| + \cdots +  |z_{n-m}|,
\end{align*}
and using the induction hypothesis, we find
$$
|z_n| =
|z_{n-1}| + |z_{n-2}| +  |z_{n-3}| + \cdots + |z_{n-m}|.
$$

Secondly assume that $m$ is odd, and as is the previous case, proceed by induction on $n\ge m-1$.
If $n=m-1$, then using Corollary~\ref{cor:comparison-of-length-mbonacci-1} several times, we have 
\begin{align*}
|z_{m-1}| &= |z_{m-2} | + 2 |z_{m-3}|  = |z_{m-2} | + |z_{m-3}| + |z_{m-4} | + 2 |z_{m-5}|  =  |z_{m-2} | + |z_{m-3}|  + \cdots + |z_1|  + 2 |z_0| \\
&= |z_{m-2} | + |z_{m-3}|  + \cdots +  |z_1| + |z_0| + |z_{-1}| + 1 
\end{align*}
since $|z_0| = 1$ and $|z_{-1}|=0$. 
Now suppose that $n\ge m$, and assume that result holds for all values less than $n$. 
From Definition~\ref{def:z-n-rec-mbonacci}, we have
\begin{align*}
|z_n| 
&= 2 |z_{n-2}| + 2 |z_{n-3}| + \cdots + 2 |z_{n-m}| + |z_{n-(m+1)}| \\
&= ( |z_{n-2}| + |z_{n-3}| + \cdots + |z_{n-m}| + |z_{n-(m+1)}| -(-1)^n)
+ |z_{n-2}| + |z_{n-3}| + \cdots + |z_{n-m}| +(-1)^n\\
&= ( |z_{n-2}| + |z_{n-3}| + \cdots + |z_{n-m}| + |z_{n-(m+1)}| +(-1)^{n-1})
+ |z_{n-2}| + |z_{n-3}| + \cdots + |z_{n-m}| +(-1)^n.
\end{align*}
The induction hypothesis allows us to conclude that
$$
|z_n| 
= |z_{n-1}| + |z_{n-2}| + |z_{n-3}| + \cdots + |z_{n-m}| +(-1)^n,
$$
as desired. 
%
\end{proof}

The following inequalities on the lengths of p-singular words will be useful later on. 

\begin{proposition}\label{pro:comparison-of-length-mbonacci}
We have the following inequalities. 
\begin{itemize}
\item[(1)] For all $0 \le n \le m-1$, 
$$
|z_n| \ge \sum_{k=0}^{n-1} |z_k| = \sum_{k=-1}^{n-1} |z_k|.
$$
\item[(2)] For all $n \ge 1$, $|z_n| \ge |z_{n-1}| +  |z_{n-2}|$.
\item[(3)] For all $n \ge 1$, $|z_{n+1}| > |z_{n}|$.
\end{itemize}
\end{proposition}
\begin{proof}
Let us prove (1) by induction on $0 \le n \le m-1$.
If $n=0$, then $|z_{-1}| = 0 \le 1 = |z_0|$.
If $n=1$, then $|z_{-1}| + |z_{0}| = 1 \le 1 = |z_1|$.
Now suppose that $2 \le n \le m-1$, and that the result is true for values less than $n$. 
By the induction hypothesis, we have
$$
\sum_{k=0}^{n-1} |z_k| = \sum_{k=0}^{n-3} |z_k| + |z_{n-2}| + |z_{n-1}| \le 2 |z_{n-2}| + |z_{n-1}|,
$$
and by Corollary~\ref{cor:comparison-of-length-mbonacci-1}, we have
$$
\sum_{k=0}^{n-1} |z_k|  \le |z_n|.
$$

Let us prove (2). 
First, suppose that $1 \le n \le m-1$. 
Then Corollary~\ref{cor:comparison-of-length-mbonacci-1} implies that $|z_{n}| = |z_{n-1}| + 2 |z_{n-2}| \ge |z_{n-1}| +  |z_{n-2}|$ since $ |z_{n-2}| \ge  |z_{-1}|  =0$. 
Suppose that $n\ge m$.
If $m$ is even, then by Corollary~\ref{cor:comparison-of-length-mbonacci-2}, we know that
$$
|z_{n}| =
|z_{n-1}| + |z_{n-2}| + |z_{n-3}| + \cdots + |z_{n-m}| \ge |z_{n-1}| + |z_{n-2}|
$$
since $|z_{n-3}|, \ldots, |z_{n-m}| \ge |z_{-1}| =0$ (when $m=2$, the inequality above is an equality).
If $m$ is odd, then by Corollary~\ref{cor:comparison-of-length-mbonacci-2}, we have 
$$
|z_{n}| =
|z_{n-1}| + |z_{n-2}| + \cdots + |z_{n-m}| + (-1)^n.
$$
When $n$ is even, then clearly $|z_n| \ge |z_{n-1}| + |z_{n-2}|$.
When $n$ is odd, then $ |z_{n-m}| - 1 \ge 0$, so we have $|z_n| \ge |z_{n-1}| + |z_{n-2}|$.

Let us show that (3) holds. 
Suppose that $n \ge 1$. 
Since $|z_{n-1}| > 0$, the result can easily be deduced as a corollary of (2).
\end{proof}

From Table~\ref{table:first-few-words-z-n}, one can observe that the first few words in two consecutive sequences of p-singular words are the same.
In the following proposition, we compare the first $m+1$ terms of the sequences $(z_n^{(m)})_{n\ge -1}$ and $(z_n^{(m+1)})_{n\ge -1}$ by showing that they are equal. 
Also notice the words differ after that.

\begin{proposition}\label{prop:comparison-p-sing-words-different-m}
For all $-1\le n \le m-1$, we have $z^{(m)}_n=z^{(m+1)}_n$.
\end{proposition}
\begin{proof}
We proceed by induction on $n$, with $-1\le n \le m-1$.
It is clear that for all $m\ge 2$, we have $z_{-1}^{(m)}=\varepsilon$ and $z_{0}^{(m)}=0$, 
so the base case is true.
Now suppose that $1 \le n \le m-1$, and that $z^{(m)}_k=z^{(m+1)}_k$ for $-1 \le k < n$. 
From Definition~\ref{def:z-n-rec-mbonacci}(1), we have
$$
z^{(m)}_n = z^{(m)}_{n-2} z^{(m)}_{n-3} \cdots z^{(m)}_1 z^{(m)}_0 n z^{(m)}_0 z^{(m)}_1 \cdots z^{(m)}_{n-3} z^{(m)}_{n-2}.
$$
Now using the induction hypothesis, we have
$$
z^{(m)}_n = z^{(m+1)}_{n-2} z^{(m+1)}_{n-3} \cdots z^{(m+1)}_1 z^{(m+1)}_0 n z^{(m+1)}_0 z^{(m+1)}_1 \cdots z^{(m+1)}_{n-3} z^{(m+1)}_{n-2},
$$
and using Definition~\ref{def:z-n-rec-mbonacci}(1) again, we get
$z^{(m)}_n=z^{(m+1)}_n$.
\end{proof}

The idea to obtain the palindromic $z$- and $c$-factorizations of the $m$-bonacci word is to mimic the reasoning in the previous case. 
Namely, we establish results similar to Propositions~\ref{pro:properties-of-singular-words},~\ref{pro:fact-Fici-Fibonacci} and~\ref{pro:prefix-Fib}.
Before getting those properties in the more general $m$-bonacci case,
a few preliminaries are necessary. 
In the following lemma, we get a formula for the p-singular word $z_n$ in terms of the morphism $\phi_m$ and the p-singular word $z_{n-1}$.

\begin{lemma}\label{lem:z-n-phi-m}
For all $n\ge 0$, 
$$
z_n
=
\begin{cases}
 0^{-1}\phi_m(z_{n-1}),  & \text{if } n \text{ is odd}; \\
\phi_m(z_{n-1})0,  & \text{if } n \text{ is even}.
\end{cases}
$$
\end{lemma}
\begin{proof}
We proceed by induction on $n\ge 0$.
If $n=0$, then $z_0=0=\phi_m(\varepsilon)0=\phi_m(z_{-1})0$. 
If $n=1$, then 
$$z_1=1=0^{-1}01=0^{-1}\phi_m(0)=0^{-1}\phi_m(z_0).$$
Now suppose that $n\ge 2$ and that the result is true for values less than $n$. 
As a first case, suppose that $2 \le n \le m-1$. In particular, $1 \le n-1\le m-2$, and we deduce from Definition~\ref{def:z-n-rec-mbonacci}(1) that
$$
z_{n-1}= z_{n-3} z_{n-4} \cdots z_1 z_0 (n-1) z_0 z_1 \cdots z_{n-4} z_{n-3}.
$$
If $n$ is even, then 
\begin{align*}
\phi_m(z_{n-1}) 0
=& \phi_m(z_{n-3}) \phi_m(z_{n-4}) \cdots \phi_m(z_1) \phi_m(z_0) 0n \phi_m(z_0) \phi_m(z_1) \cdots \phi_m(z_{n-4}) \phi_m(z_{n-3}) 0 \\
=& \phi_m(z_{n-3}) 00^{-1} \phi_m(z_{n-4}) \cdots \phi_m(z_1) 00^{-1} \phi_m(z_0) 0n 00^{-1} \phi_m(z_0) \phi_m(z_1) 0 \cdots 0^{-1} \phi_m(z_{n-4}) \\
&\phi_m(z_{n-3}) 0.
\end{align*}
By the induction hypothesis, we obtain
\begin{align*}
\phi_m(z_{n-1}) 0
&= z_{n-2} z_{n-3} \cdots z_2 z_1 0n 0 z_1 z_2 \cdots z_{n-3} z_{n-2}
\\
&= z_{n-2} z_{n-3} \cdots z_2 z_1 z_0 n z_0 z_1 z_2 \cdots z_{n-3} z_{n-2},
\end{align*}
and the last equality holds because $z_0=0$.
From Definition~\ref{def:z-n-rec-mbonacci}(1), we have $\phi_m(z_{n-1}) 0=z_n$.
If $n$ is odd, then 
\begin{align*}
0^{-1} \phi_m(z_{n-1}) 
=& 0^{-1} \phi_m(z_{n-3}) \phi_m(z_{n-4}) \cdots \phi_m(z_1) \phi_m(z_0) 0n \phi_m(z_0) \phi_m(z_1) \cdots \phi_m(z_{n-4}) \phi_m(z_{n-3}) \\
=& 0^{-1} \phi_m(z_{n-3}) \phi_m(z_{n-4}) 0  \cdots \phi_m(z_1) 00^{-1} \phi_m(z_0) 0n 00^{-1} \phi_m(z_0) \phi_m(z_1) 0 \cdots \phi_m(z_{n-4}) 0 \\
&0^{-1}  \phi_m(z_{n-3}).
\end{align*}
By the induction hypothesis, we obtain
\begin{align*}
0^{-1} \phi_m(z_{n-1})
&= z_{n-2} z_{n-3} \cdots z_2 z_1 0n 0 z_1 z_2 \cdots z_{n-3} z_{n-2}
\\
&= z_{n-2} z_{n-3} \cdots z_2 z_1 z_0 n z_0 z_1 z_2 \cdots z_{n-3} z_{n-2},
\end{align*}
and the last equality is true because $z_0=0$.
From Definition~\ref{def:z-n-rec-mbonacci}(1), we have $0^{-1}\phi_m(z_{n-1})=z_n$.

Now suppose that $n=m$ (this implies that $m=n\ge 2$). 
By Definition~\ref{def:z-n-rec-mbonacci}(1), we have
$$
z_{m-1}= z_{m-3} z_{m-4} \cdots z_1 z_0 (m-1) z_0 z_1 \cdots z_{m-4} z_{m-3}.
$$
If $n=m$ is even, then 
\begin{align*}
\phi_m(z_{m-1}) 0
=& \phi_m(z_{m-3}) \phi_m(z_{m-4}) \cdots \phi_m(z_1) \phi_m(z_0) 0 \phi_m(z_0) \phi_m(z_1) \cdots \phi_m(z_{m-4}) \phi_m(z_{m-3}) 0 \\
=& \phi_m(z_{m-3}) 00^{-1} \phi_m(z_{m-4}) \cdots \phi_m(z_1) 00^{-1} \phi_m(z_0) 0 00^{-1} \phi_m(z_0) \phi_m(z_1) 0 \cdots 0^{-1} \phi_m(z_{m-4}) \\
&\phi_m(z_{m-3}) 0.
\end{align*}
By the induction hypothesis and since $z_0=0$, we obtain
\begin{align*}
\phi_m(z_{m-1}) 0
&= z_{m-2} z_{m-3} \cdots z_2 z_1 00 z_1 z_2 \cdots z_{m-3} z_{m-2}
\\
&= z_{m-2} z_{m-3} \cdots z_2 z_1 z_0 z_0 z_1 z_2 \cdots z_{m-3} z_{m-2}.
\end{align*}
From Definition~\ref{def:z-n-rec-mbonacci}(2), we have $\phi_m(z_{m-1}) 0=z_m$.
If $n=m$ is odd, then 
\begin{align*}
0^{-1} \phi_m(z_{m-1}) 
=& 0^{-1} \phi_m(z_{m-3}) \phi_m(z_{m-4}) \cdots \phi_m(z_1) \phi_m(z_0) 0 \phi_m(z_0) \phi_m(z_1) \cdots \phi_m(z_{m-4}) \phi_m(z_{m-3}) \\
=& 0^{-1} \phi_m(z_{m-3}) \phi_m(z_{m-4}) 0  \cdots \phi_m(z_1) 00^{-1} \phi_m(z_0) 0 00^{-1} \phi_m(z_0) \phi_m(z_1) 0 \cdots \phi_m(z_{m-4}) 0 \\
&0^{-1}  \phi_m(z_{m-3}).
\end{align*}
By the induction hypothesis and since $z_0=0$, we obtain
\begin{align*}
0^{-1} \phi_m(z_{m-1})
&= z_{m-2} z_{m-3} \cdots z_2 z_1 0 0 z_1 z_2 \cdots z_{m-3} z_{m-2}
\\
&= z_{m-2} z_{m-3} \cdots z_2 z_1 z_0 z_0 z_1 z_2 \cdots z_{m-3} z_{m-2}.
\end{align*}
From Definition~\ref{def:z-n-rec-mbonacci}(2), we have $0^{-1}\phi_m(z_{m-1})=z_m$.

Finally, assume that $n\ge m +1$. 
By Definition~\ref{def:z-n-rec-mbonacci}(2), we have
$$
z_{n-1}= z_{n-3} z_{n-4} \cdots z_{n-1-(m-1)} z_{n-1-m} z_{n-1-(m+1)} z_{n-1-m} z_{n-1-(m-1)}  \cdots z_{n-4} z_{n-3}.
$$
If $n$ is even, then 
\begin{align*}
\phi_m(z_{n-1}) 0
=& \phi_m(z_{n-3}) \phi_m(z_{n-4}) \cdots \phi_m(z_{n-1-(m-1)}) \phi_m(z_{n-1-m}) \\ &\phi_m(z_{n-1-(m+1)}) \phi_m(z_{n-1-m}) \phi_m(z_{n-1-(m-1)})  \cdots \phi_m(z_{n-4}) \phi_m(z_{n-3}) 0.
\end{align*}
Inserting $00^{-1}$ where needed (places where to insert it differ when $m$ is even or odd) and using the induction hypothesis , we obtain
$$
\phi_m(z_{n-1}) 0
= z_{n-2} z_{n-3} \cdots z_{n-(m-1)} z_{n-m} z_{n-(m+1)} z_{n-m} z_{n-(m-1)}  \cdots z_{n-3} z_{n-2}.
$$
From Definition~\ref{def:z-n-rec-mbonacci}(2), we have $\phi_m(z_{n-1}) 0=z_n$ as desired.
If $n$ is odd, then 
\begin{align*}
0^{-1} \phi_m(z_{m-1}) 
=& 0^{-1} \phi_m(z_{n-3}) \phi_m(z_{n-4}) \cdots \phi_m(z_{n-1-(m-1)}) \phi_m(z_{n-1-m}) \\ &\phi_m(z_{n-1-(m+1)}) \phi_m(z_{n-1-m}) \phi_m(z_{n-1-(m-1)})  \cdots \phi_m(z_{n-4}) \phi_m(z_{n-3}).
\end{align*}
As before, inserting $00^{-1}$ where needed and making use of the induction hypothesis, we get
$$
0^{-1} \phi_m(z_{n-1})
=z_{n-2} z_{n-3} \cdots z_{n-(m-1)} z_{n-m} z_{n-(m+1)} z_{n-m} z_{n-(m-1)}  \cdots z_{n-3} z_{n-2}.
$$
From Definition~\ref{def:z-n-rec-mbonacci}(2), we have $0^{-1}\phi_m(z_{n-1})=z_n$.
This ends the proof. 
\end{proof}

Using the previous lemma, we are able to prove Lemma~\ref{lem:z-n-mbonacci}.

\begin{proof}[Proof of Lemma~\ref{lem:z-n-mbonacci}]
For the sake of simplicity, let us drop the exponent $(m)$ in this proof. 
We equivalently show that, if $n\ge 0$ is odd,
$$
h_{n-1} h_{n-3} \cdots h_{2} h_{0} z_n = h_{n} h_{n-2} \cdots h_{3} h_{1}
$$
and if $n\ge 0$ is even,
$$
h_{n-1} h_{n-3} \cdots h_{3} h_{1} z_n = h_{n} h_{n-2} \cdots h_{2} h_{0}.
$$
We proceed by induction on $n\ge 0$. 
If $n=0$, then $\varepsilon \cdot z_0=0=h_0$ holds.
If $n=1$, then $h_0 z_1=01=h_1$.
Now suppose that $n\ge 2$ and that the result is true for values less than $n$. 

If $n$ is odd, then $n-1$ is even and the induction hypothesis yields
$$
h_{n-2} h_{n-4} \cdots h_{3} h_{1} z_{n-1} = h_{n-1} h_{n-3} \cdots h_{2} h_{0}.
$$ 
Applying $\phi_m$ on both sides, we get 
$$
h_{n-1} h_{n-3} \cdots h_{4} h_{2} \phi_m(z_{n-1}) = h_{n} h_{n-2} \cdots h_{3} h_{1}.
$$ 
We may now insert $00^{-1}$ before $\phi_m(z_{n-1})$ in the left-hand side of the previous equality to obtain
$$
h_{n-1} h_{n-3} \cdots h_{4} h_{2} 00^{-1} \phi_m(z_{n-1}) = h_{n} h_{n-2} \cdots h_{3} h_{1}.
$$ 
We conclude by using the fact that $h_0=0$ and $0^{-1} \phi_m(z_{n-1})=z_n$ thanks to Lemma~\ref{lem:z-n-phi-m}.

If $n$ is even, then $n-1$ is odd and the induction assumption gives
$$
h_{n-2} h_{n-4} \cdots h_{2} h_{0} z_{n-1} = h_{n-1} h_{n-3} \cdots h_{3} h_{1}.
$$ 
Applying $\phi_m$ on both sides and appending a letter $0$, we obtain 
$$
h_{n-1} h_{n-3} \cdots h_{3} h_{1} \phi_m(z_{n-1}) 0 = h_{n} h_{n-2} \cdots h_{4} h_{2} 0.
$$ 
We end this case by using the fact that $h_0=0$ and $\phi_m(z_{n-1})0=z_n$ thanks to Lemma~\ref{lem:z-n-phi-m}.
\end{proof}

The following result matches Proposition~\ref{pro:properties-of-singular-words}(4) in the Fibonacci case. 

\begin{proposition}\label{pro:z-n-not-factor-mbonacci}
For all $n\ge 0$, $z_n$ is not a factor of $z_{n+1}$.
\end{proposition}
\begin{proof}
Observe first that the case $m=2$ is covered using Propositions~\ref{pro:properties-of-singular-words}(4) and~\ref{pro:f-n-z-n}.
So we can suppose that $m\ge 3$, and we proceed by induction on $n\ge 0$. 
The result can be checked by hand for $0\le n \le 2$ since $z_0=0$, $z_1=1$, and $z_2=020$ (see Definition~\ref{def:z-n-rec-mbonacci}).
Suppose that $n\ge 3$ and assume that $z_k$ is not a factor of $z_{k+1}$ for $0 \le k \le n-1$. 
We show it still holds for $k=n$, i.e., $z_n$ is not a factor of $z_{n+1}$.
Proceed by contradiction and suppose that $z_n$ is a factor of $z_{n+1}$.
We divide the proof into two cases according to the parity of $n$.

\textbf{Case 1}. Suppose that $n$ is odd.
From Lemma~\ref{lem:z-n-phi-m}, we know that $z_n=0^{-1}\phi_m(z_{n-1})$ and $z_{n+1}=\phi_m(z_{n})0$.
We prove that $0z_n=00^{-1}\phi_m(z_{n-1})=\phi_m(z_{n-1})$ is a factor of $z_{n+1}=\phi_m(z_{n})0$. 
By hypothesis, there exist words $x,y\in\{0,1,\ldots,m-1\}^*$ such that $z_{n+1}=0xz_ny$. If $x=\varepsilon$, then the statement is true. If $|x|=\ell>0$, write $x=x_1 \cdots x_\ell$ with $x_i \in\{0,1,\ldots,m-1\}$ for all $1\le i \le \ell$. By definition of the morphism $\phi_m$, we find $x_\ell=0$ since $z_n$ starts with a positive letter, which ends the intermediate result. Now, we claim that $\phi_m(z_{n-1})$ is in fact a factor of $\phi_m(z_{n})$.
First, using Proposition~\ref{pro:z-n-first-last-letters}, 
$$
z_{n-1}=0u0 \quad \text{and} \quad z_{n}=1v1
$$ 
for non-empty words $u,v\in\{0,1,\ldots,m-1\}^+$ (recall that $n-1\ge 2$). 
Thus $\phi_m(z_{n-1})$ (resp., $\phi_m(z_{n})$) starts and ends with $01$ (resp., $02$). 
Consequently, $\phi_m(z_{n-1})$ is a factor of $\phi_m(z_{n})$, as claimed. 
From Lemma~\ref{lem:strange-technical-lemma-v2}, $z_{n-1}$ is a factor of $z_n$, which contradicts the induction hypothesis.

\textbf{Case 2}. Assume that $n$ is even. 
From Lemma~\ref{lem:z-n-phi-m}, we know that $z_n=\phi_m(z_{n-1})0$ and $z_{n+1}=0^{-1}\phi_m(z_{n})$.
By hypothesis, $z_n=\phi_m(z_{n-1})0$ is a factor of $0z_{n+1}=00^{-1}\phi_m(z_{n})=\phi_m(z_{n})$.
Thus, $\phi_m(z_{n-1})$ is in fact a factor of $\phi_m(z_{n})$.
Using Proposition~\ref{pro:z-n-first-last-letters}, 
$$
z_{n-1}=1u1 \quad \text{and} \quad z_{n}=0v0
$$ 
for non-empty words $u,v\in\{0,1,\ldots,m-1\}^+$ (recall that $n-1\ge 2$). 
From Lemma~\ref{lem:strange-technical-lemma-v2}, $z_{n-1}$ is a factor of $z_n$, which again contradicts the induction hypothesis.
\end{proof}

The following result is the counterpart to Proposition~\ref{pro:properties-of-singular-words}(5) in the Fibonacci case. 

\begin{proposition}\label{pro:z-n-not-factor-prod-mbonacci}
For all $n\ge 1$, $z_n$ is not a factor of the product $\prod_{k=0}^{n-1} z_{k}$.
\end{proposition}
\begin{proof}
For the sake of simplicity, we define, for all $n\ge 1$, 
$$
P(n) = \prod_{k=0}^{n-1} z_{k}.
$$
Observe first that the case $m=2$ follows from Propositions~\ref{pro:properties-of-singular-words}(5) and~\ref{pro:f-n-z-n}.
So we can suppose that $m\ge 3$. 
To prove the result, we proceed by induction on $n\ge 1$. 

If $1\le n\le m-1$, then by Proposition~\ref{pro:comparison-of-length-mbonacci}, we have $|z_n| \ge \sum_{k=0}^{n-1} |z_k|$. 
If the inequality is strict, then we are done. If we actually have an equality, then $z_{n-1}$ would be a factor of $z_n$, which contradicts Proposition~\ref{pro:z-n-not-factor-mbonacci}. 

Suppose that $n\ge m$ and assume that $z_i$ is not a factor of $P(i)$ for $1 \le i \le n-1$. 
We show it still holds for $i=n$, i.e., $z_n$ is not a factor of $P(n)$.
Proceed by contradiction and suppose that $z_n$ is a factor of $P(n)$.
We divide the proof into two cases according to the parity of $n$.


\textbf{Case 1}. Suppose that $n$ is odd.
From Lemma~\ref{lem:z-n-phi-m}, we get 
\begin{align}
P(n)
&= z_0 z_1 \cdots z_{n-2} z_{n-1} \nonumber \\
&= (\phi_m(z_{-1})0)(0^{-1}\phi_m(z_0)) \cdots (0^{-1}\phi_m(z_{n-3}))(\phi_m(z_{n-2})0) \nonumber \\
&= \phi_m(z_{-1} z_0 \cdots z_{n-3} z_{n-2}) 0 \nonumber \\
&= \phi_m(P(n-1)) 0. \label{eq:z-n-not-factor-prod-mbonacci-1}
\end{align}
By hypothesis, $z_n=0^{-1}\phi_m(z_{n-1})$ is a factor of $P(n)=\phi_m(P(n-1)) 0$, so $0z_n=00^{-1}\phi_m(z_{n-1})=\phi_m(z_{n-1})$ is also a factor of $0P(n)=0\phi_m(P(n-1)) 0$ (the reasoning is similar to the one developed in the previous proof).
We claim that $\phi_m(z_{n-1})$ is in fact a factor of $\phi_m(P(n-1))$.
First, using Proposition~\ref{pro:z-n-first-last-letters}, we have
$$
z_{n-1} = 0u0 \quad \text{and} \quad P(n-1)=0v1
$$
for two non-empty words $u,v\in\{0,1,\ldots,m-1\}^+$.
Consequently, $\phi_m(z_{n-1})$ starts and ends with $01$, and $\phi_m(P(n-1))$ starts with $01$ and ends with $02$. 
Thus, $\phi_m(z_{n-1})$ is a factor of $\phi_m(P(n-1))$, as expected.
From Lemma~\ref{lem:strange-technical-lemma-v2}, $z_{n-1}$ is a factor of $P(n-1)$, which contradicts the induction hypothesis.

\textbf{Case 2}. Assume that $n$ is even.
From Lemma~\ref{lem:z-n-phi-m}, we get 
\begin{align}
P(n)
&= z_0 z_1 \cdots z_{n-2} z_{n-1} \nonumber \\
&= (\phi_m(z_{-1})0)(0^{-1}\phi_m(z_0)) \cdots (\phi_m(z_{n-3})0)(0^{-1}\phi_m(z_{n-2})) \nonumber \\
&= \phi_m(z_{-1} z_0 \cdots z_{n-3} z_{n-2})  \nonumber \\
&= \phi_m(P(n-1)). \label{eq:z-n-not-factor-prod-mbonacci-2}
\end{align}
By hypothesis, $z_n=\phi_m(z_{n-1})0$ is a factor of $P(n)=\phi_m(P(n-1))$, so $\phi_m(z_{n-1})$ is also a factor of $\phi_m(P(n-1))$.
Using Proposition~\ref{pro:z-n-first-last-letters}, we have
$$
z_{n-1} = 1u1 \quad \text{and} \quad P(n-1)=0v0
$$
for two non-empty words $u,v\in\{0,1,\ldots,m-1\}^+$.
From Lemma~\ref{lem:strange-technical-lemma-v2}, $z_{n-1}$ is a factor of $P(n-1)$, which also contradicts the induction hypothesis.
\end{proof}

\begin{proposition}\label{pro:z-n-do-not-contain}
For all $0 \le n \le m-1$, the words $z_{-1}, z_{0}, z_{1}, \ldots, z_{n-1}$ do not contain the letter $n$. 
\end{proposition}
\begin{proof}
If $n=0$, then $z_{-1}$ is the empty word, and we are done. 
If $n=1$, each of the words $z_{-1}=\varepsilon$ and $z_{0}=0$ does not contain the letter $1$.
Let $2 \le n \le m-1$.
Then the words $(z_i^{(n)})_{i\ge -1}$ are well defined (see Definition~\ref{def:z-n-rec-mbonacci}).
Iteratively applying Proposition~\ref{prop:comparison-p-sing-words-different-m}, we obtain $z_i^{(n)}=z_i^{(m)}$ for all $-1\le i \le n-1$.
Since $z_i^{(n)}$ is a word defined over the alphabet $\{0,1,\ldots,n-1\}$ for any $i\ge -1$, the conclusion follows. 
%
%
\end{proof}

The following result compares the prefixes of $z_{n-1}$ to suffixes of $z_n$. Its proof follows the same lines as the proof of Lemma~\ref{lem:common-suffix-prefix-Fib} in the Fibonacci case.

\begin{lemma}\label{lem:common-suffix-prefix-mbonacci}
For all $n\ge 0$, the only suffix of $z_{n-1}$ that is also a prefix of $z_n$ is the empty word. 
\end{lemma}
\begin{proof}
We proceed by induction on $n\ge 0$.
From Definition~\ref{def:z-n-rec-mbonacci}, we have $z_{-1}=\varepsilon$, $z_0=0$ and $z_1=1$, so the result can be checked by hand for $n\in\{0,1\}$.

Now suppose that $n\ge 2$, and that the only suffix of $z_{k-1}$ that is also a prefix of $z_k$ is the empty word, for all $k\in\{0, 1, \ldots, n-1\}$.
We show that the result still holds for $k=n$. 
Proceed by contradiction and suppose there exists a word $x\in\{0,1,\ldots,m-1\}^{*}$ which is a non-empty suffix of $z_{n-1}$ and a non-empty prefix of $z_n$. We have $1 \le |x| \le |z_{n-1}|$. 
Using Definition~\ref{def:z-n-rec-mbonacci}, $z_n$ starts and ends with $z_{n-2}$. 

If $1 \le |x| \le |z_{n-2}|$, then $x$ is a prefix of $z_{n-2}$ (recall that $x$ is a prefix of $z_n$).
Consequently, $x^R$ is a non-empty suffix of $z_{n-2}$ and a non-empty prefix of $z_{n-1}$. This contradicts the inductive assumption. 

If $|z_{n-2}| \le |x| \le |z_{n-1}|$, then $z_{n-2}$ is a prefix of $x$ (recall that $x$ is a prefix of $z_n$). 
In particular, $z_{n-2}$ is a factor of $x$, and also a factor of $z_{n-1}$ (recall that $x$ is a suffix of $z_{n-1}$).
This contradicts Proposition~\ref{pro:z-n-not-factor-mbonacci}.
\end{proof}

\subsection{Two particular factorizations of the $m$-bonacci word}

In this section, we study two different factorizations of the $m$-bonacci word in terms of p-singular words (see Propositions~\ref{pro:mbonacci-fact-z-n} and~\ref{pro:mbonacci-fact-p-m-Q-n}), extending Proposition~\ref{pro:fact-Fici-Fibonacci}.
The first one is similar to the factorization~\eqref{eq:fact1} of the Fibonacci word given in Proposition~\ref{pro:fact-Fici-Fibonacci}.
To see this, simply put~\eqref{eq:fact1} and Proposition~\ref{pro:f-n-z-n} altogether.

\begin{proposition}\label{pro:mbonacci-fact-z-n}
We have the following factorization of the $m$-bonacci word
$$
\mbonacci = \prod_{n \ge 0} z_n.
$$ 
\end{proposition}
\begin{proof}
For all $n\ge 0$, set $P(n)=\prod_{k=0}^{n-1} z_{k}$ (when $n=0$, $P(0)$ is the empty word).
To prove the statement, we show two things: 
\begin{itemize}
\item[(1)] For all $n\ge 1$, $|P(n)| > |P(n-1)|$,
\item[(2)] $(P(n))_{n\ge 0}$ is a sequence of prefixes of $\mbonacci$.
\end{itemize}
Then, the mentioned factorization easily follows. 
For all $n\ge 1$, we trivially have 
$$
|P(n)| = |P(n-1)| + |z_{n-1}| > |P(n-1)|,
$$
since $n\ge 1$ implies $|z_{n-1}|>0$.
Thus (1) is proved. 
For (2), we proceed by induction on $n\ge 0$. 
The $m$-bonacci word $\mbonacci$ starts with $01$, so it is clear that $P(n)$ is a prefix of $\mbonacci$ for $n\in\{0,1,2\}$.
Suppose that $n\ge 3$ and that $P(n-1)$ is a prefix of $\mbonacci$. 
The proof is again divided into two parts, according to the parity of $n$. 

\textbf{Case 1}. Suppose that $n$ is odd. 
From~\eqref{eq:z-n-not-factor-prod-mbonacci-1} (which is valid for any odd $n\ge 0$), we know that 
$P(n) = \phi_m(P(n-1)) 0$, and using Proposition~\ref{pro:z-n-first-last-letters}, $P(n-1)$ ends with $1$. 
By the induction hypothesis, $P(n-1)$ is a prefix of $\mbonacci$ ending with $1$. 
Thus, there exists an infinite word $\boldsymbol{z}$ over $\{0,1, \ldots, m-1\}$ such that $\mbonacci = P(n-1) 0 \boldsymbol{z}$. 
Since $\mbonacci$ is a fixed point of $\phi_m$, we get
$$
\mbonacci = \phi_m(\mbonacci) = \phi_m(P(n-1)) 01 \phi_m(\boldsymbol{z}) = P(n)1 \phi_m(\boldsymbol{z}),
$$ 
showing that $P(n)$ is also a prefix of $\mbonacci$.

\textbf{Case 2}. Assume that $n$ is even. 
From~\eqref{eq:z-n-not-factor-prod-mbonacci-2} (which is valid for any even $n\ge 0$), we already have $P(n)= \phi_m(P(n-1))$.
By the induction hypothesis, there exists an infinite word $\boldsymbol{z}$ over $\{0,1, \ldots, m-1\}$ such that $\mbonacci = P(n-1) \boldsymbol{z}$.
Since $\mbonacci$ is a fixed point of $\phi_m$, we get
$$
\mbonacci = \phi_m(\mbonacci) = \phi_m(P(n-1)) \phi_m(\boldsymbol{z}) = P(n) \phi_m(\boldsymbol{z}),
$$ 
as desired.
\end{proof}

The factorization of the $m$-bonacci word in Proposition~\ref{pro:mbonacci-fact-p-m-Q-n} is similar to the factorization~\eqref{eq:fact2} of the Fibonacci word given in Proposition~\ref{pro:fact-Fici-Fibonacci}.
We first need some notations. 

\begin{definition}\label{def:mbonacci-particular-prefix-and-factors}
Let $p_m$ be the finite word over $\{0,1, \ldots, m-1\}$ defined by 
$$
z_0 z_1 \cdots z_{m-4} z_{m-3} z_{m-2} z_{m-3} z_{m-4} \cdots z_1 z_0 (m-1).
$$
For all $n\ge m-2$, define the finite word $Q(n)$ over $\{0,1, \ldots, m-1\}$ by
$$
Q(n) = z_{n-(m-1)} z_{n-(m-2)} \cdots z_{n-2} z_{n-1} z_n z_{n-1} z_{n-2} \cdots  z_{n-(m-2)} z_{n-(m-1)}.
$$
Note that the word $Q(n)$ is centered at $z_n$. 
\end{definition}

\begin{example}
If $m=2$, $p_2=z_0(2-1)=01$ and $Q(n)=z_{n-1} z_n z_{n-1}$ for all $n\ge 0$.
When $m=3$, we find $p_3=z_0z_1z_0(3-1)=0102$, and for all $n\ge 1$, we have $Q(n)=z_{n-2} z_{n-1} z_n z_{n-1} z_{n-2}$.
\end{example}

\begin{proposition}\label{pro:mbonacci-fact-p-m-Q-n}
We have the following factorization of the $m$-bonacci word
$$
\mbonacci =
p_m \cdot \prod_{n \ge m-2} Q(n).
$$
\end{proposition}
\begin{proof}
From Proposition~\ref{pro:mbonacci-fact-z-n}, we get
\begin{equation}\label{eq:mbonacci-second-fact-1}
\mbonacci = \prod_{n \ge 0} z_n 
= z_0 \cdot z_1 \cdots z_{m-2} \cdot  z_{m-1} \cdot  \prod_{n \ge m} z_n.
\end{equation}
Using Definition~\ref{def:z-n-rec-mbonacci}(1), we have
$$
z_{m-1} = z_{m-3} z_{m-4} \cdots z_1 z_0 (m-1) z_0 z_1 \cdots z_{m-4} z_{m-3},
$$
and for all $n\ge m$, Definition~\ref{def:z-n-rec-mbonacci}(2) shows that
$$
z_n= z_{n-2} z_{n-3} \cdots z_{n-(m-1)} z_{n-m} z_{n-(m+1)} z_{n-m} z_{n-(m-1)}  \cdots z_{n-3} z_{n-2}.
$$
Plugging these equalities into~\eqref{eq:mbonacci-second-fact-1}, we find
\begin{align*}
\mbonacci =& z_0 \cdot z_1 \cdots z_{m-2} \cdot (z_{m-3} z_{m-4} \cdots z_1 z_0 (m-1) z_0 z_1 \cdots z_{m-4} z_{m-3}) \\ &\cdot  \prod_{n \ge m} z_{n-2} z_{n-3} \cdots z_{n-(m-1)} z_{n-m} z_{n-(m+1)} z_{n-m} z_{n-(m-1)}  \cdots z_{n-3} z_{n-2} \\
=& p_m  \cdot z_0 z_1 \cdots z_{m-4} z_{m-3} \\
 &\cdot  \prod_{n \ge m} z_{n-2} z_{n-3} \cdots z_{n-(m-1)} z_{n-m} z_{n-(m+1)} z_{n-m} z_{n-(m-1)}  \cdots z_{n-3} z_{n-2}
\end{align*}
using Definition~\ref{def:mbonacci-particular-prefix-and-factors}. 
Since $z_{-1}=\varepsilon$, we finally get
\[
\mbonacci = p_m \prod_{n \ge m-2} z_{n-(m-1)} z_{n-(m-2)} \cdots z_{n-2} z_{n-1} z_n z_{n-1} z_{n-2} \cdots  z_{n-(m-2)} z_{n-(m-1)} = p_m \cdot \prod_{n \ge m-2} Q(n).
\qedhere
\]
\end{proof}

In the following proposition, we get a particular factorization of the prefix $p_m$ of the $m$-bonacci word $\mbonacci$.
This factorization is a step forward to obtain the palindromic $c$-factorization of the $m$-bonacci word $\mbonacci$.

\begin{proposition}\label{pro:mbonacci-p-m-fact}
The word $p_m$ can be factorized as
$$
p_m = z_0 \cdot 1 \cdot z_0 \cdot 2 \cdot (z_0 z_1 z_0) \cdot 3 \cdots (z_0 z_1 \cdots z_{m-4} z_{m-3} z_{m-4} \cdots z_1 z_0) \cdot (m-1).
$$
In particular, this factorization contains $2m-2$ factors and all of them are palindromes. 
Moreover, if this factorization is written as
$$
p_m = q_1 q_2 \cdots q_{2m-2}, 
$$ 
then, for all $1 \le k \le 2m-2$ and for any infinite word $\boldsymbol{w}$, $q_k$ is the longest palindromic prefix of $q_k q_{k+1} \cdots q_{2m-2}\boldsymbol{w}$ with a previous occurrence in $p_{m}\boldsymbol{w}=q_1 q_2 \cdots q_{2m-2}\boldsymbol{w}$, or if this prefix does not exist, the factor $q_k$ is a single letter.
\end{proposition}
\begin{proof}
To prove this result, we proceed by induction on $m\ge 2$. 
To avoid any confusion, from Definition~\ref{def:mbonacci-particular-prefix-and-factors}, write 
$$
p_m =
z^{(m)}_0 z^{(m)}_1 \cdots z^{(m)}_{m-4} z^{(m)}_{m-3} z^{(m)}_{m-2} z^{(m)}_{m-3} z^{(m)}_{m-4} \cdots z^{(m)}_1 z^{(m)}_0 (m-1).
$$
The case $m=2$ is easily checked for we have $p_2=01=z_0^{(2)}\cdot 1$. 
Now suppose that $m\ge 3$ and assume that the result holds for values less than $m$. 
Let us prove the first part of the statement.
Using Definition~\ref{def:z-n-rec-mbonacci}(1) to rewrite $ z^{(m)}_{m-2}$, we first have
\begin{align}
p_m =&
z^{(m)}_0 z^{(m)}_1 \cdots z^{(m)}_{m-4} z^{(m)}_{m-3} (z^{(m)}_{m-4} z^{(m)}_{m-5} \cdots z^{(m)}_1 z^{(m)}_0 (m-2) z^{(m)}_0 z^{(m)}_1 \cdots z^{(m)}_{m-5} z^{(m)}_{m-4}) \nonumber \\ &z^{(m)}_{m-3} z^{(m)}_{m-4} \cdots z^{(m)}_1 z^{(m)}_0 (m-1) \nonumber \\
=& z^{(m)}_0 z^{(m)}_1 \cdots z^{(m)}_{m-5} z^{(m)}_{m-4} z^{(m)}_{m-3}  z^{(m)}_{m-4} z^{(m)}_{m-5} \cdots z^{(m)}_1 z^{(m)}_0 (m-2) \label{eq:p-m-fact} \\
&z^{(m)}_0 z^{(m)}_1 \cdots z^{(m)}_{m-5} z^{(m)}_{m-4} z^{(m)}_{m-3} z^{(m)}_{m-4} z^{(m)}_{m-5} \cdots z^{(m)}_1 z^{(m)}_0 (m-1). \nonumber
\end{align}
From Proposition~\ref{prop:comparison-p-sing-words-different-m}, the two finite words
$$
z^{(m)}_0 z^{(m)}_1 \cdots z^{(m)}_{m-5} z^{(m)}_{m-4} z^{(m)}_{m-3} z^{(m)}_{m-4} z^{(m)}_{m-5} \cdots z^{(m)}_1 z^{(m)}_0 (m-2)
$$
and 
$$
z^{(m-1)}_0 z^{(m-1)}_1 \cdots z^{(m-1)}_{m-5} z^{(m-1)}_{m-4} z^{(m-1)}_{m-3} z^{(m-1)}_{m-4} z^{(m-1)}_{m-5} \cdots z^{(m-1)}_1 z^{(m-1)}_0 (m-2)
$$
are equal since $z^{(m)}_k = z^{(m-1)}_k$ for all $-1 \le k \le m-2$.
Since the latest word is $p_{m-1}$ by Definition~\ref{def:mbonacci-particular-prefix-and-factors}, we deduce that
$$
p_m = p_{m-1} z^{(m)}_0 z^{(m)}_1 \cdots z^{(m)}_{m-5} z^{(m)}_{m-4} z^{(m)}_{m-3} z^{(m)}_{m-4}  z^{(m)}_{m-5} \cdots z^{(m)}_1 z^{(m)}_0 (m-1).
$$ 
By the induction hypothesis, we know that
\begin{align*}
p_{m-1} =& z^{(m-1)}_0 \cdot 1 \cdot z^{(m-1)}_0 \cdot 2 \cdot (z^{(m-1)}_0 z^{(m-1)}_1 z^{(m-1)}_0) \cdot 3 \cdots \\
&(z^{(m-1)}_0 z^{(m-1)}_1 \cdots z^{(m-1)}_{m-5} z^{(m-1)}_{m-4} z^{(m-1)}_{m-5} \cdots z^{(m-1)}_1 z^{(m-1)}_0) \cdot (m-2).
\end{align*}
Proposition~\ref{prop:comparison-p-sing-words-different-m} finally gives
\begin{align*}
p_m =& z^{(m)}_0 \cdot 1 \cdot z^{(m)}_0 \cdot 2 \cdot (z^{(m)}_0 z^{(m)}_1 z^{(m)}_0) \cdot 3 \cdots (z^{(m)}_0 z^{(m)}_1 \cdots z^{(m)}_{m-5} z^{(m)}_{m-4} z^{(m)}_{m-5} \cdots z^{(m)}_1 z^{(m)}_0) \cdot (m-2)  \\
& (z^{(m)}_0 z^{(m)}_1 \cdots z^{(m)}_{m-5}  z^{(m)}_{m-4} z^{(m)}_{m-3} z^{(m)}_{m-4} z^{(m)}_{m-5} \cdots z^{(m)}_1 z^{(m)}_0) \cdot (m-1),
\end{align*}
as expected. 
Moreover, using the induction hypothesis, this factorization contains $(2\cdot (m-1)-2) + 2=2 m - 2$ factors, which are all palindromes. 
Note that we have 
$$
q_{2m-3} = z^{(m)}_0 z^{(m)}_1 \cdots z^{(m)}_{m-5}  z^{(m)}_{m-4} z^{(m)}_{m-3} z^{(m)}_{m-4} z^{(m)}_{m-5} \cdots z^{(m)}_1 z^{(m)}_0
$$
and $q_{2m-2} = m-1$. 
This ends the proof of the first part of the statement.

Let us show that the second part of the statement also holds. 
The proof is divided into three cases according to the value of the index of the considered factor $q_k$. 

\textbf{Case 1}. Suppose that $1\le k \le 2m-4$. 
For all infinite word $\boldsymbol{w}$, $q_k$ is the longest palindromic prefix of $q_k q_{k+1} \cdots q_{2m-2}\boldsymbol{w}$ with a previous occurrence in $p_{m}\boldsymbol{w}=q_1 q_2 \cdots q_{2m-2}\boldsymbol{w}$, or if this prefix does not exist, $q_k$ is limited to a single letter.
Indeed, by the induction hypothesis and since $q_{2m-3}q_{2m-2}\boldsymbol{w}$ is a particular infinite word, $q_k$ is the longest palindromic prefix of $q_k q_{k+1} \cdots q_{2m-4}(q_{2m-3}q_{2m-2}\boldsymbol{w})$ with a previous occurrence in $p_{m-1}(q_{2m-3}q_{2m-2}\boldsymbol{w})=p_{m}\boldsymbol{w}=q_1 q_2 \cdots q_{2m-2}\boldsymbol{w}$, or if this prefix does not exist, $q_k$ is limited to a single letter..

\textbf{Case 2}. Assume that $k = 2m-3$, and let $\boldsymbol{w}$ be any infinite word. 
Looking at~\eqref{eq:p-m-fact}, we get
$$
p_m \boldsymbol{w} 
= q_1 q_2 \cdots q_{2m-4} q_{2m-3} q_{2m-2} \boldsymbol{w}
= q_{2m-3} (m-2) q_{2m-3} (m-1) \boldsymbol{w}.
$$
Using Proposition~\ref{pro:z-n-do-not-contain}, we see that $q_{2m-3}$ is the longest palindromic prefix of $q_{2m-3} q_{2m-2} \boldsymbol{w}$ that has already occurred in $p_m \boldsymbol{w}$.

\textbf{Case 3}. Suppose that $k = 2m-2$, and let $\boldsymbol{w}$ be any infinite word.
Proposition~\ref{pro:z-n-do-not-contain} shows that $ m-1=q_{2m-2}$ does not appear previously in $p_m \boldsymbol{w}$. 
Hence, $q_{2m-2} = m-1$ also satisfies the second part of the statement.
%
%
%
\end{proof}

Since the idea is to adopt the same strategy as in the previous case, we define a sequence of specific prefixes of the $m$-bonacci word $\mbonacci$. 
This definition gives the sequence of prefixes of Definition~\ref{def:prefix} in the Fibonacci case as proved in Remark~\ref{rk:mbonacci-Fibonacci-prefix-g-n}.

\begin{definition}\label{def:mbonacci-prefix-g-n}
For all $n\ge m-2$, define 
$$
g_n:= p_m \prod_{m-2 \le k \le n-1} Q(k), 
$$
where the words $p_m$ and $(Q(k))_{k\ge m-2}$ are given in Definition~\ref{def:mbonacci-particular-prefix-and-factors}.
Notice that $g_{m-2} = p_m  \cdot \varepsilon = p_m$.
From Proposition~\ref{pro:mbonacci-fact-p-m-Q-n}, also observe that, for all $n\ge m-2$, we have 
$$
\mbonacci = g_n \cdot Q(n) \cdot \prod_{ k\ge n+1 } Q(k).
$$
\end{definition}

\begin{remark}\label{rk:mbonacci-Fibonacci-prefix-g-n}
Let $(g'_n)_{n\ge 2}$ denote the sequence of words of Definition~\ref{def:prefix}.
For all $n\ge 1$, Proposition~\ref{pro:f-n-z-n} gives
$$
g'_{n+1}
= 010 \prod_{2 \le k \le n} \hat{f}_{k-1} \hat{f}_{k} \hat{f}_{k-1}
= 010 \prod_{2 \le k \le n} z_{k-2} z_{k-1} z_{k-2} 
= 01 \prod_{0 \le k \le n-1} z_{k-1} z_{k} z_{k-1}
= g_n. 
$$
This shows that Definition~\ref{def:mbonacci-prefix-g-n} agrees with Definition~\ref{def:prefix} when $m=2$. 
\end{remark}

As in the Fibonacci case (see Proposition~\ref{pro:prefix-Fib}), any word $g_n$ can be written using p-singular words. 

\begin{proposition}\label{pro:mbonacci-prefix-g-n}
For all $n\ge m-1$, we have 
$$
g_n = z_0 z_1 \cdots z_{n-m} \cdot Q(n) \cdot z_{n-m}.
$$
\end{proposition}
\begin{proof}
We proceed by induction on $n\ge m-1$.
For the base case $n=m-1$, Definitions~\ref{def:mbonacci-particular-prefix-and-factors} and~\ref{def:mbonacci-prefix-g-n} give
\begin{align*}
g_{m-1} =& p_m Q(m-2) \\
=& z_0 z_1 \cdots z_{m-4} z_{m-3} z_{m-2} z_{m-3} z_{m-4} \cdots z_1 z_0 (m-1) \\
&(z_{-1} z_{0} \cdots z_{m-4} z_{m-3} z_{m-2} z_{m-3} z_{m-4} \cdots  z_{0} z_{-1}).
\end{align*}
By Definition~\ref{def:z-n-rec-mbonacci}(1) ($1 \le m-1$), we have
$$
z_{m-1} = z_{m-3} z_{m-4} \cdots z_1 z_0 (m-1) z_0 z_1 \cdots z_{m-4} z_{m-3}.
$$
Using Definition~\ref{def:mbonacci-particular-prefix-and-factors}, we thus have
\begin{align*}
g_{m-1} 
&= z_0 z_1 \cdots z_{m-4} z_{m-3} z_{m-2} z_{m-1} z_{m-2} z_{m-3} z_{m-4} \cdots  z_{0} z_{-1} \\
&= z_{-1} (z_0 z_1 \cdots z_{m-4} z_{m-3} z_{m-2} z_{m-1} z_{m-2} z_{m-3} z_{m-4} \cdots  z_1 z_{0}) z_{-1} \\
& = z_{-1} \cdot Q(m-1) \cdot z_{-1},
\end{align*}
as expected. 
%
%
Assume that $n\ge m-1$, and suppose the result holds up to $n$ and we show it still holds for $n+1$. 
Using Definition~\ref{def:mbonacci-prefix-g-n}, we have $g_{n+1} = g_n  Q(n)$.
By the induction hypothesis, we get
$$
g_{n+1}
= z_0 z_1 \cdots z_{n-m} \cdot Q(n) \cdot z_{n-m} \cdot Q(n).
$$
Rewriting $Q(n)$ using Definition~\ref{def:mbonacci-particular-prefix-and-factors}, we find
\begin{align*}
g_{n+1}
=& z_0 z_1 \cdots z_{n-m} \cdot (z_{n-(m-1)} z_{n-(m-2)} \cdots z_{n-2} z_{n-1} z_n z_{n-1} z_{n-2} \cdots  z_{n-(m-2)} z_{n-(m-1)}) \cdot z_{n-m} \\ 
&\cdot (z_{n-(m-1)} z_{n-(m-2)} \cdots z_{n-2} z_{n-1} z_n z_{n-1} z_{n-2} \cdots  z_{n-(m-2)} z_{n-(m-1)}).
\end{align*}
Since $n+1 \ge m$, from Definition~\ref{def:z-n-rec-mbonacci}(2), we have 
$$
z_{n+1}=z_{n-1} z_{n-2} \cdots z_{n+1-(m-1)} z_{n+1-m} z_{n+1-(m+1)} z_{n+1-m} z_{n+1-(m-1)}  \cdots z_{n-2} z_{n-1},
$$ 
and we deduce that
\begin{align*}
g_{n+1} 
&= z_0 z_1 \cdots z_{n-m} z_{n+1-m} z_{n+1-(m-1)} \cdots z_{n-2} z_{n-1} z_n z_{n+1} z_n z_{n-1} z_{n-2} \cdots  z_{n+1-(m-1)} z_{n+1-m} \\
&=z_0 z_1 \cdots z_{n-m} z_{n+1-m} ( z_{n+1-(m-1)}  z_{n+1-(m-2)} \cdots z_{n-1} z_n z_{n+1} z_n z_{n-1}  \cdots  z_{n+1-(m-1)} ) z_{n+1-m}.
\end{align*}
Consequently, from Definition~\ref{def:mbonacci-particular-prefix-and-factors}, we obtain 
$$
g_{n+1} 
= z_0 z_1 \cdots z_{n-m} z_{n+1-m} \cdot Q(n+1) \cdot z_{n+1-m},
$$
which ends the proof. 
\end{proof}

\subsection{The palindromic $z$-factorization of the $m$-bonacci word}

In this section, we obtain the palindromic $z$-factorization of the $m$-bonacci word.

\begin{lemma}\label{lem:pal-preimages}
  Let $p$ be a non-empty palindromic factor of $\mbonacci$.
\begin{itemize}
  \item If $p$ begins with the letter $0$, then $p = \phi_m(p')0$, where $p'$ is a palindromic factor of $\mbonacci$.
  \item If $p$ begins with the letter $a\neq 0$, then $p = 0^{-1}\phi_m(p')$, where $p'$ is a palindromic factor of $\mbonacci$.
\end{itemize}
\end{lemma}

\begin{proof}
First, let us write $\mbonacci=up\boldsymbol{v}$ where $u$ (resp., $\boldsymbol{v}$) is a finite (resp., infinite) word over $\{0,1,\ldots, m-1\}$.
  By definition, we have 
  \begin{equation}\label{eq:pal-preimages}
  \mbonacci=\phi_m(\mbonacci) 
  \quad \text{ and } \quad
  |\phi_m(up)|>|up|.
  \end{equation}

  The proof is by induction on $|p|$.  The result is certainly true when
  $p$ is a single letter (for instance, combine Propositions~\ref{pro:mbonacci-fact-p-m-Q-n} and~\ref{pro:mbonacci-p-m-fact}), so suppose $|p|>1$.

  \textbf{Case 1a}.  Suppose $p$ begins with $00$.  If $p=00$, then $p
  = \phi_m(m-1)0$, as required (observe that $m-1$ is indeed a palindromic factor of $\mbonacci$: for instance, make use of Propositions~\ref{pro:mbonacci-fact-p-m-Q-n} and~\ref{pro:mbonacci-p-m-fact}).  Suppose $p=00q00$.  By the induction
  hypothesis, we have $0q0 = \phi_m(q')0$, where $q'$ is a palindromic factor of $\mbonacci$.
  We get $p = 0\phi_m(q')00 = \phi_m((m-1)q'(m-1))0$.
  It is clear that $p'=(m-1)q'(m-1)$ is a palindrome. 
  Let us show that $p'$ is also a factor of $\mbonacci$, then we are done.
  From~\eqref{eq:pal-preimages}, $p= \phi_m(p')0$ being a factor of $up$ implies that it is also a factor of $\phi_m(up)0$. 
  By Lemma~\ref{lem:strange-technical-lemma-v2}, $p'$ is a factor of $up$, so of $\mbonacci$.

  \textbf{Case 1b}.  Suppose $p$ begins with $0a$, where $a \in\{1,2,\ldots,m-1\}$.
  Then $p=0aqa0$.  By the induction
  hypothesis, we have $aqa = 0^{-1}\phi_m(p')$, where $p'$ is a palindromic factor of $\mbonacci$.
  We get $p = 00^{-1}\phi_m(p')0 = \phi_m(p')0$, as required.

  \textbf{Case 2}.  Suppose $p$ begins with $a$, where $a \in\{1,2,\ldots,m-1\}$.
  Then $p=aqa$, where $q$ is a palindromic factor of $\mbonacci$ that begins with $0$.  By
  the induction hypothesis, we have $q = \phi_m(q')0$ for a palindromic factor $q'$ of $\mbonacci$.  We get $p =
  0^{-1}0a\phi_m(q')0a = 0^{-1}\phi_m((a-1)q'(a-1))$.
  Again, it is easy to see that $p'=(a-1)q'(a-1)$ is a palindrome. It remains to prove that $p'$ is a factor of $\mbonacci$.
  From~\eqref{eq:pal-preimages}, we deduce that $p= 0^{-1}\phi_m(p')$ being a factor of $up$ implies that it is a factor of $\phi_m(up)$ too. 
  Thus, $\phi_m(p')$ is a factor of $\phi_m(up)$.
  By Lemma~\ref{lem:strange-technical-lemma-v2}, $p'$ is a factor of $up$, so also of $\mbonacci$.
\end{proof}

\begin{lemma}\label{lem:pal-prefixes-z-n-mbonacci}
Let $n\ge -1$. 
The set of palindromic prefixes $\mathcal{P}(z_n)$ of $z_n$ is 
$$
\mathcal{P}(z_n) 
=
\begin{cases}
\{z_{-1}, z_1, z_3, \ldots, z_{n-2}, z_n\},  & \text{if } n \text{ is odd}; \\
\{z_{-1}, z_0, z_2, \ldots, z_{n-2}, z_n\},  & \text{if } n \text{ is even}.
\end{cases}
$$
\end{lemma}
\begin{proof}
  The proof is by induction on $n$.  The result is clearly true for
  $n=-1,0,1$.  Suppose $n>1$ is even, the case where $n$ is odd is
  similar. First, let $p$ be a palindromic prefix
  of $z_n$.  Since $n$ is even, the word $z_n$, and hence $p$, begins
  with $0$ by Proposition~\ref{pro:z-n-first-last-letters}.  
  By Lemma~\ref{lem:pal-preimages} (and also Proposition~\ref{pro:mbonacci-fact-z-n}), we have $p =
  \phi_m(p')0$ where $p'$ is a palindrome.  By Lemma~\ref{lem:z-n-phi-m}, the word $p'$ is a
  palindromic prefix of $z_{n-1}$.  By the induction hypothesis, we
  have $p' \in \mathcal{P}(z_{n-1})$; i.e., $p'=z_i$ for some
  $i\in\{-1,1,3,\ldots,n-3,n-1\}$.  
  By Lemma~\ref{lem:z-n-phi-m} again, we have $p=\phi_m(z_i)0=z_{i+1} \in
  \mathcal{P}(z_n)$.
We have just showed that $\mathcal{P}(z_n) \subset \{z_{-1}, z_0, z_2, \ldots, z_{n-2}, z_n\}$.
  
Let us prove that the other inclusion holds too. 
We clearly have  $z_{-1}, z_n \in \mathcal{P}(z_n)$. 
Now let $i \in \{0,2,\ldots,n-2\}$. 
By Lemma~\ref{lem:z-n-phi-m}, $z_i = \phi_m(z_{i-1})0$. 
By the induction hypothesis, we know that $z_{i-1} \in \mathcal{P}(z_{n-1})$, i.e., there exists a non-empty word $w\in \{0,1,\ldots,m-1\}^*$ such that $z_{n-1}=z_{i-1}w$. 
Using Lemma~\ref{lem:z-n-phi-m} again, we get $z_n= \phi_m(z_{n-1})0
= \phi_m(z_{i-1})\phi_m(w)0$. 
By definition of $\phi_m$, we have $\phi_m(w)0=0w'$ with $w'\in \{0,1,\ldots,m-1\}^*$.
As a consequence, we find $z_n = \phi_m(z_{i-1})0w'=z_iw'$ and $z_i\in \mathcal{P}(z_n)$, as expected.
\end{proof}

The last result of this section establishes the $z$-factorization of the $m$-bonacci word $\mbonacci$ in terms of p-singular words.

\begin{theorem}\label{thm:pal-z-fact-mbonacci}
The palindromic $z$-factorization of the $m$-bonacci word $\mbonacci$ is 
$$
pz(\mbonacci)=(z_0, z_1, z_2, \ldots).
$$ 
\end{theorem}
\begin{proof}
If $m=2$, one simply has to combine Theorem~\ref{thm:pal-z-fact-fibonacci} and Proposition~\ref{pro:f-n-z-n}.
Now assume that $m\ge 3$.
Let $pz(\mbonacci)=(t_0, t_1, t_2, \ldots)$ be the palindromic $z$-factorization of the $m$-bonacci word $\mbonacci$.
We proceed by induction on $n\ge 0$ to show that $t_n=z_n$. 
From Proposition~\ref{pro:mbonacci-fact-z-n}, we have
\begin{equation}\label{eq:mbonacci-z-fact}
\mbonacci= \prod_{n \ge 0} z_n  = 0 \cdot  1 \cdot  020 \cdot z_3 \cdot z_4 \cdots,
\end{equation}
and we see that the first two factors of the $z$-factorization of the $m$bonacci word $\mbonacci$ are $t_0=0=z_0$ and $t_1=1=z_1$. 
Now suppose that $n\ge 2$. 
From~\eqref{eq:mbonacci-z-fact}, we deduce that 
$$
\mbonacci
= \prod_{n \ge 0} z_n
= \left( \prod_{0 \le k \le n-1} z_k \right) z_n z_{n+1} \cdots.
$$ 
To prove that $t_n=z_n$, we need to show that every palindromic prefix of $z_n$ which is different from $z_n$ is a factor of
$$
P(n) = \prod_{k=0}^{n-1} z_{k}.
$$
This is clear from Lemma~\ref{lem:pal-prefixes-z-n-mbonacci}.
\end{proof}

\subsection{The palindromic $c$-factorization of the $m$-bonacci word}

In the following lemma, which is similar to Lemma~\ref{lem:two-occurrences}, recall that we start indexing words at $0$.

\begin{lemma}\label{lem:two-occurrences-mbonacci}
Let $n\ge m-1$. 
There are exactly two occurrences of the factor $z_n$ inside the word $g_{n+1}$: one at position $\sum_{k=0}^{n-1} | z_k | $, the other at position $\sum_{k=0}^{n+1} | z_k | $.
\end{lemma}
\begin{proof}
First consider the case $m=2$, and let $n\ge 1$. 
Using the notation introduced in Remark~\ref{rk:mbonacci-Fibonacci-prefix-g-n}, Lemma~\ref{lem:two-occurrences} shows that there are exactly two occurrences of the factor $\hat{f}_{n+1}=z_{n}$ inside the word $g'_{n+2}=g_{n+1}$, one occurring at position $\sum_{k=1}^{n} | \hat{f}_{k} |=\sum_{k=0}^{n-1} | z_k | $, the other at position $\sum_{k=1}^{n+2} | \hat{f}_{k} |=\sum_{k=0}^{n+1} | z_k | $, as desired. 

Assume that $m\ge 3$.
Let $n\ge m-1$. 
Using Definition~\ref{def:mbonacci-particular-prefix-and-factors} and Proposition~\ref{pro:mbonacci-prefix-g-n}, let us write 
\begin{align*}
g_{n+1} 
&= z_0 z_1 \cdots z_{n+1-m} Q(n+1) z_{n+1-m}  \\
&= z_0 z_1 \cdots z_{n+1-m} (z_{n+2-m} z_{n+3-m} \cdots z_{n-1} z_{n} z_{n+1} z_{n} z_{n-1} \cdots  z_{n+3-m} z_{n+2-m}) z_{n+1-m} \\
&=u z_{n} z_{n+1} z_{n} v
\end{align*} 
with $u = z_0 z_1 \cdots z_{n-1}$ and $v=z_{n-1} z_{n-2} \cdots z_{n+1-m}$. 
Thanks to this factorization, we immediately see that $z_n$ occurs at least twice as a factor of $g_{n+1}$: one starting at position $|u|=\sum_{k=0}^{n-1} | z_{k} |$, the other beginning at position $|u z_{n} z_{n+1}|=\sum_{k=0}^{n+1} | z_{k} |$. 
We now show that there are no other occurrences of $z_{n}$ as a factor of $g_{n+1}$. There are several cases to consider. 

\textbf{Case 1}. The word $z_{n}$ cannot be a factor of $u$, otherwise it contradicts Proposition~\ref{pro:z-n-not-factor-prod-mbonacci}.

\textbf{Case 2}. The word $z_{n}$ cannot be a factor of $z_{n+1}$, otherwise it contradicts Proposition~\ref{pro:z-n-not-factor-mbonacci}.

\textbf{Case 3}. If the word $z_{n}$ were a factor of $v$, then $(z_{n})^{R}=z_{n}$ would be a factor of $v^R=z_{n+1-m} \cdots z_{n-2} z_{n-1}$, since the p-singular words are palindromes.
This is impossible due to Proposition~\ref{pro:z-n-not-factor-prod-mbonacci}.

\textbf{Case 4}. Suppose that $z_n$ is a factor of $uz_{n}$, overlapping $u$ and $z_{n}$.
Using Corollary~\ref{cor:comparison-of-length-mbonacci-2} ($n\ge m-1$), we know that
$$
|z_{n}| - 1 \le |z_{n-1}| + |z_{n-2}| + \cdots + |z_{n-m}|.  
$$
Consequently, $z_{n}$ is a factor of $z_{n-m} z_{n-(m-1)} \cdots z_{n-2} z_{n-1}z_{n}$, overlapping $z_{n-m} z_{n-(m-1)} \cdots z_{n-2} z_{n-1}$ and $z_{n}$.

If $z_{n}$ starts somewhere within $z_{n-k}$, with $k\in\{2, 3, \ldots m\}$, or if $z_{n}$ starts with the first letter of $z_{n-1}$, then, in each case, $z_{n-1}$ is a factor of $z_{n}$, which contradicts Proposition~\ref{pro:z-n-not-factor-mbonacci}.
Therefore the occurrence of $z_n$ must start after the first letter of $z_{n-1}$ in $z_{n-m} z_{n-(m-1)} \cdots z_{n-2} z_{n-1}z_{n}$, i.e., there exist a non-empty suffix $x$ of $z_{n-1}$ and a non-empty prefix $y$ of $z_{n}$ such that 
$$z_{n} = xy.$$
Observe that, in this case, $x$ is also a non-empty prefix of $z_n$. This contradicts Lemma~\ref{lem:common-suffix-prefix-mbonacci}.

\textbf{Case 5}. Suppose that $z_n$ is a factor of $z_{n}z_{n+1}$, overlapping $z_{n}$ and $z_{n+1}$.
There exist a non-empty suffix $x'$ of $z_{n}$ and a non-empty prefix $y'$ of $z_{n+1}$ such that 
$$z_{n} = x'y'.$$
In this case, notice that $y'$ is also a non-empty suffix of $z_n$. This violates Lemma~\ref{lem:common-suffix-prefix-mbonacci}.

\textbf{Case 6}. Suppose that $z_n$ is a factor of $z_{n+1}z_{n}$, overlapping $z_{n+1}$ and $z_{n}$.
Since the p-singular words are palindromes, we obtain that $(z_{n})^{R}=z_{n}$ is a factor of $(z_{n+1}z_{n})^{R}=z_{n}z_{n+1}$, overlapping $z_{n}$ and $z_{n+1}$.
As in the fifth case, we raise a contradiction.

\textbf{Case 7}. Suppose that $z_n$ is a factor of $z_{n}v$, overlapping $z_{n}$ and $v$.
Then $(z_{n})^{R}=z_{n}$ is a factor of $(z_{n}v)^{R}=z_{n+1-m} z_{n-m} \cdots z_{n-1} z_{n}$, overlapping (at least) $z_{n-1}$ and $z_{n}$. 
In the view of the fourth case, we also reach a contradiction.
\end{proof}

The following result is the counterpart to Proposition~\ref{pro:nothing-to-add}.

\begin{proposition}\label{pro:nothing-to-add-mbonacci}
Let $n\ge m-1$. 
Let $w$ be a non-empty common finite prefix of the infinite words 
$$
\boldsymbol{u_{n+1,1}} = z_{n-(m-1)} Q(n+1) Q(n+2) \cdots
$$
and
$$
\boldsymbol{u_{n+1,2}} = Q(n+2) Q(n+3) \cdots.
$$ 
Then $Q(n+1) w$ is not a palindrome.
\end{proposition}
\begin{proof}
If $m=2$, the result directly follows from
Propositions~\ref{pro:nothing-to-add} and~\ref{pro:f-n-z-n}.  Suppose
that $m\ge 3$.  We proceed by induction on $n$.
  First, suppose $n=m-1$.  Then $\boldsymbol{u_{m,1}} = 0 Q(m) Q(m+1)
  \cdots$ and $\boldsymbol{u_{m,2}} = Q(m+1) Q(m+2) \cdots$.  Note
  that $Q(m)$ begins with $z_1=1$ and $Q(m+1)$ begins with $z_2=020$.
  The only possibility for $w$ is $w=0$ and in this case $Q(m)0$ is
  not a palindrome, as required.

  We now suppose that the result holds for $\boldsymbol{u_{n,1}}$ and
  $\boldsymbol{u_{n,2}}$.  We proceed by contradiction and suppose
  that $Q(n+1)w$ is a palindrome.

  We first observe that by
  Definition~\ref{def:mbonacci-particular-prefix-and-factors} and
  Lemma~\ref{lem:z-n-phi-m} we have $Q(i) = \phi_m(Q(i-1))0$ if $Q(i)$
  begins with $0$ and $Q(i) = 0^{-1}\phi_m(Q(i-1))$ if $Q(i)$ begins with $1$. 
  Thus either
  $$\boldsymbol{u_{n+1,1}} =
  \phi_m(z_{n-m})00^{-1}\phi_m(Q(n))\phi_m(Q(n+1))0 \cdots =
  \phi_m(\boldsymbol{u_{n,1}})$$ and
  $$\boldsymbol{u_{n+1,2}} = \phi_m(Q(n+1))00^{-1}\phi_m(Q(n+2)) \cdots =
  \phi_m(\boldsymbol{u_{n,2}});$$ or
  $$\boldsymbol{u_{n+1,1}} =
  0^{-1}\phi_m(z_{n-m})\phi_m(Q(n))00^{-1}\phi_m(Q(n+1)) \cdots =
  0^{-1}\phi_m(\boldsymbol{u_{n,1}})$$ and
  $$\boldsymbol{u_{n+1,2}} = 0^{-1}\phi_m(Q(n+1))\phi_m(Q(n+2))0 \cdots =
  0^{-1}\phi_m(\boldsymbol{u_{n,2}}).$$

  By Definition~\ref{def:mbonacci-prefix-g-n}, observe that $Q(n+1)w$ is a factor of $\mbonacci$.
  Now using Lemma~\ref{lem:pal-preimages}, we have either $Q(n+1)w
  = \phi_m(Q(n)w')0=Q(n+1)0^{-1}\phi_m(w')0$ or $Q(n+1)w = 0^{-1}\phi_m(Q(n)w')=Q(n+1)\phi_m(w')$, where
  $Q(n)w'$ is a palindrome and $w'$ is a common prefix of
  $\boldsymbol{u_{n,1}}$ and $\boldsymbol{u_{n,2}}$.  This contradicts
  the induction hypothesis.  We conclude that $Q(n+1)w$ is not a
  palindrome, as required.
\end{proof}

In this last result, we obtain the $c$-factorization of the $m$-bonacci word $\mbonacci$ in terms of p-singular words.

\begin{theorem}\label{thm:pal-c-fact-mbonacci}
Let $pc(\mbonacci)=(c_{-m}, c_{-(m-1)}, c_{-(m-2)}, \ldots)$ denote the palindromic $c$-factorization of the $m$-bonacci word $\mbonacci$. 
Then, the first $2m-2$ factors are given by the factorization of $p_m$ emphasized in Proposition~\ref{pro:mbonacci-p-m-fact}, and, for all $n\ge m-2$, $c_n = Q(n)$.
\end{theorem}
\begin{proof}
Consider first the case where $m=2$. 
From Theorem~\ref{thm:pal-c-fact-fibonacci}, we know that $c_{-2}=0$, $c_{-1}=1$, $c_0=0$ and, for all $n\ge 1$, 
$$
c_n = \hat{f}_{n}  \hat{f}_{n+1} \hat{f}_{n}.
$$
We clearly have $c_0= \varepsilon \cdot 0 \cdot \varepsilon = z_{-1} z_0 z_{-1} = Q(0)$, and for all $n\ge 1$, we also get $c_n = z_{n-1} z_n z_{n-1}= Q(n)$ using Proposition~\ref{pro:f-n-z-n}.
Assume now that $m\ge 3$.
From Definition~\ref{def:mbonacci-particular-prefix-and-factors},  Propositions~\ref{pro:mbonacci-fact-p-m-Q-n} and~\ref{pro:mbonacci-p-m-fact}, we have 
\begin{align}
\mbonacci 
&= p_m \prod_{n \ge m-2} Q(n) \nonumber \\
&=  z_0 z_1 \cdots z_{m-4} z_{m-3} z_{m-2} z_{m-3} z_{m-4} \cdots z_1 z_0 (m-1) \prod_{n \ge m-2} Q(n) \label{eq:pal-c-fact-mbonacci-1} \\
&= z_0 \cdot 1 \cdot z_0 \cdot 2 \cdot (z_0 z_1 z_0) \cdot 3 \cdots (z_0 z_1 \cdots z_{m-4} z_{m-3} z_{m-4} \cdots z_1 z_0) \cdot (m-1) \cdot \prod_{n \ge m-2} Q(n). \label{eq:pal-c-fact-mbonacci-2}
\end{align}
Using the definition of the $c$-palindromic factorization and looking at~\eqref{eq:pal-c-fact-mbonacci-2}, the first $(2m-2)$ factors 
$c_{-m}, c_{-(m-1)}, \ldots, c_{m-3}$ 
of $pc(\mbonacci)$ are given by the factorization of $p_m$ emphasized in Proposition~\ref{pro:mbonacci-p-m-fact} (recall the second part of Proposition~\ref{pro:mbonacci-p-m-fact}: each $q_k=c_{k-m-1}$, for $1\le k \le 2m-2$, is the longest palindromic prefix of $q_kq_{k+1}\cdots q_{2m-2} \cdot \prod_{n \ge m-2} Q(n)$ that has already occurred in $\mbonacci = p_m \prod_{n \ge m-2} Q(n)$, or if this prefix does not exist, $q_k$ is a single letter).

For the second part of the result, proceed by induction on $n\ge m-2$.
If $n=m-2$, we must find the factor $c_{m-2}$ of the palindromic $c$-factorization of $\mbonacci$.
Using Definition~\ref{def:mbonacci-particular-prefix-and-factors}, we have
$$
Q(m-2) = z_{-1} z_{0} \cdots z_{m-4} z_{m-3} z_{m-2} z_{m-3} z_{m-4} \cdots  z_{0} z_{-1},
$$
and from~\eqref{eq:pal-c-fact-mbonacci-1}, we get
\begin{align}
\mbonacci 
=&
z_0 z_1 \cdots z_{m-4} z_{m-3} z_{m-2} z_{m-3} z_{m-4} \cdots z_1 z_0 (m-1) (z_{-1} z_{0} \cdots z_{m-4} z_{m-3} z_{m-2} z_{m-3} z_{m-4} \cdots  z_{0} z_{-1}) \nonumber\\ 
&\prod_{n \ge m-1} Q(n) \label{eq:pal-c-fact-mbonacci-3} \\
=&
Q(m-2) \cdot (m-1) \cdot Q(m-2) \cdot \prod_{n \ge m-1} Q(n). \nonumber
\end{align}
On the one hand, observe that $Q(m-1)$ starts with $z_0$ by Definition~\ref{def:mbonacci-particular-prefix-and-factors}, and $z_0 \neq m-1$. On the other hand, $Q(m-2)$ does not contain the letter $m-1$ by Proposition~\ref{pro:z-n-do-not-contain}.
Thus, the longest palindrome occurring before is $c_{m-2}=Q(m-2)$.

If $n=m-1$, we must find the factor $c_{m-1}$ of the palindromic $c$-factorization of $\mbonacci$.
Using Definition~\ref{def:mbonacci-particular-prefix-and-factors}, we have
$$
Q(m-1) = z_{0} z_{1} \cdots z_{m-3} z_{m-2} z_{m-1} z_{m-2} z_{m-3} \cdots z_{1}  z_{0}.
$$
Starting from~\eqref{eq:pal-c-fact-mbonacci-3} and using Definition~\ref{def:z-n-rec-mbonacci}(1), we have
\begin{align*}
\mbonacci 
=&
z_0 z_1 \cdots z_{m-4} z_{m-3} z_{m-2} (z_{m-3} z_{m-4} \cdots z_1 z_0 (m-1) z_0 z_1 \cdots z_{m-4} z_{m-3}) z_{m-2} z_{m-3} z_{m-4} \cdots  z_{0} \\ 
&\prod_{n \ge m-1} Q(n) \\
=& z_0 z_1 \cdots z_{m-4} z_{m-3} z_{m-2} \cdot  z_{m-1}  \cdot z_{m-2} z_{m-3} z_{m-4} \cdots  z_{0} \cdot Q(m-1) \cdot \prod_{n \ge m} Q(n).
\end{align*}
On the one hand, $Q(m)$ starts with $z_1$ by Definition~\ref{def:mbonacci-particular-prefix-and-factors}, and $z_1 \neq z_0$.
On the other hand, Proposition~\ref{pro:z-n-do-not-contain} shows that $m-1$ never occurs in $z_{-1},z_0,z_1,\ldots,z_{m-2}$.
So the longest palindrome occurring before is $c_{m-1}=Q(m-1)$. 

For the induction step, suppose $n\ge m-1$ and assume that, for all $1\le k \le n$, we have $c_k = Q(k)$.
We show it is still true for $k=n+1$. 
On the one hand, using the induction hypothesis, Proposition~\ref{pro:mbonacci-p-m-fact} and finally Definition~\ref{def:mbonacci-prefix-g-n}, we have
\begin{eqnarray}
\mbonacci
&=& c_{-m} c_{-(m-1)} \cdots c_{m-3} \left( \prod_{m-2 \le k \le n} c_k \right) c_{n+1} c_{n+2} \cdots \nonumber  \\
&=& p_m \left( \prod_{m-2 \le k \le n} Q(k) \right) c_{n+1} c_{n+2} \cdots \nonumber  \\
&=& g_{n+1} c_{n+1} c_{n+2} \cdots, \label{eq:pal-c-fact-mbonacci-ind-step-1}
\end{eqnarray}
and the goal is to find the next factor of the palindromic $c$-factorization of $\mbonacci$, i.e., the word $c_{n+1}$.
On the other hand, using Proposition~\ref{pro:mbonacci-fact-p-m-Q-n} first, then Definition~\ref{def:mbonacci-prefix-g-n} and finally Proposition~\ref{pro:mbonacci-prefix-g-n}, we get
\begin{eqnarray}
\mbonacci
&=& p_m \left( \prod_{m-2 \le k \le n} Q(k)\right) Q(n+1) Q(n+2) \cdots \nonumber \\
&=& g_{n+1} Q(n+1) Q(n+2)  \cdots \nonumber \\
&=& (z_0 z_1 \cdots z_{n+1-m} \cdot Q(n+1) \cdot z_{n+1-m}) Q(n+1) Q(n+2)  \cdots. \label{eq:pal-c-fact-mbonacci-ind-step-2}
\end{eqnarray}
Comparing~\eqref{eq:pal-c-fact-mbonacci-ind-step-1} and~\eqref{eq:pal-c-fact-mbonacci-ind-step-2}, we see that $|c_{n+1}| \ge |Q(n+1) |$ since $Q(n+1)$ is a palindrome occurring before in $g_{n+1}$.
Therefore, there exists a word $w \in \{0,1,\ldots,m-1\}^{*}$ such that $c_{n+1} = Q(n+1) w$.
We claim that $w$ is in fact the empty word. 
For the following argument, recall that Definition~\ref{def:mbonacci-particular-prefix-and-factors} gives
$$
Q(n+1)=z_{n+2-m} z_{n+3-m} \cdots z_{n-1} z_{n} z_{n+1} z_{n} z_{n-1} \cdots  z_{n+3-m} z_{n+2-m}.
$$
By Lemma~\ref{lem:two-occurrences-mbonacci}, we know that there are exactly two occurrences of $z_{n}$ in $g_{n+1}$: one starting at position $\sum_{k=0}^{n-1} | z_{k} | $, the other at position $\sum_{k=0}^{n+1} | z_{k} | $. 

\textbf{Case 1}. Let us deal with the first occurrence of $z_n$ in $g_{n+1}$.
In this case, $w$ must be a common prefix of the infinite words 
$$
\boldsymbol{u_{n+1,1}} = z_{n+1-m} Q(n+1) Q(n+2) \cdots
$$
and
$$
\boldsymbol{u_{n+1,2}} = Q(n+2) Q(n+3) \cdots.
$$ 
But by Proposition~\ref{pro:nothing-to-add-mbonacci}, we know that $Q(n+1) w$ is not a palindrome unless $w$ is the empty word.

\textbf{Case 2}. Let us examine the second occurrence of $z_n$ in $g_{n+1}$.
In this case, the suffix $z_{n} z_{n+1} z_{n} z_{n-1} \cdots  z_{n+3-m} z_{n+2-m}$ of $Q(n+1)$ starts at position $\sum_{k=0}^{n+1} | z_{k} | $ in $g_{n+1}$.
In particular, $z_{n+1}$ is a prefix of the infinite word 
\begin{align*}
z_{n-1} &z_{n-2} \cdots  z_{n+3-m} z_{n+2-m} z_{n+1-m} Q(n+1) Q(n+2) \cdots \\
=& z_{n-1} z_{n-2} \cdots  z_{n+3-m} z_{n+2-m} z_{n+1-m}
(z_{n+2-m} z_{n+3-m} \cdots z_{n-1} z_{n} z_{n+1} z_{n} z_{n-1} \cdots  z_{n+3-m} z_{n+2-m}) \\
&Q(n+2) Q(n+3) \cdots.
\end{align*}
Since $n+1\ge m$, Definition~\ref{def:z-n-rec-mbonacci} gives 
$$
z_{n+1}= z_{n-1} z_{n-2} \cdots z_{n+2-m-} z_{n+1-m} z_{n-m} z_{n+1-m} z_{n+2-m}  \cdots z_{n-2} z_{n-1}.
$$
We now show that 
\begin{align}\label{eq:pal-c-fact-mbonacci-ind-step-3}
|z_{n-1} z_{n-2} \cdots  z_{n+3-m} z_{n+2-m} z_{n+1-m}
z_{n+2-m} z_{n+3-m} \cdots z_{n-2} z_{n-1}| 
< |z_{n+1}|,
\end{align}
and
\begin{align}\label{eq:pal-c-fact-mbonacci-ind-step-4}
|z_{n+1}|
\le |z_{n-1} z_{n-2} \cdots  z_{n+3-m} z_{n+2-m} z_{n+1-m}
z_{n+2-m} z_{n+3-m} \cdots z_{n-2} z_{n-1} z_{n}|.
\end{align}
Let us show~\eqref{eq:pal-c-fact-mbonacci-ind-step-3}. 
We have 
\begin{align*}
&|z_{n+1}| - |z_{n-1} z_{n-2} \cdots  z_{n+3-m} z_{n+2-m} z_{n+1-m}
z_{n+2-m} z_{n+3-m} \cdots z_{n-2} z_{n-1}|  \\
&= 2 \sum_{k=n+1-m}^{n-1} |z_{k}| + |z_{n-m}| - \left( 2 \sum_{k=n+2-m}^{n-1} |z_{k}| + |z_{n+1-m}| \right) \\
&= |z_{n+1-m}| + |z_{n-m}| >0
\end{align*}
since $n+1\ge m$ implies that $|z_{n+1-m}|>0$. 
Let us prove~\eqref{eq:pal-c-fact-mbonacci-ind-step-4}.
We get 
\begin{align*}
&|z_{n-1} z_{n-2} \cdots  z_{n+3-m} z_{n+2-m} z_{n+1-m}
z_{n+2-m} z_{n+3-m} \cdots z_{n-2} z_{n-1} z_n| - |z_{n+1}|  \\
&= 2 \sum_{k=n+2-m}^{n-1} |z_{k}| + |z_{n+1-m}| + |z_n| - \left( 2 \sum_{k=n+1-m}^{n-1} |z_{k}| + |z_{n-m}| \right) \\
&= |z_{n}| - |z_{n+1-m}| - |z_{n-m}| \ge 0.
\end{align*}
Indeed, to see this, we make use of Corollary~\ref{cor:comparison-of-length-mbonacci-2}. 
If $m$ is even, then the result is clear. 
If $m$ is odd, then $|z_{n}| \ge |z_{n+1-m}| + |z_{n-m}|$ since $|z_{n-1}| + |z_{n-2}| + \cdots + |z_{n+2-m}| + (-1)^n\ge 0$.
As a consequence of~\eqref{eq:pal-c-fact-mbonacci-ind-step-3} and~\eqref{eq:pal-c-fact-mbonacci-ind-step-4}, $z_{n+1}$ must end with a non-empty prefix of $z_n$, which contradicts Lemma~\ref{lem:common-suffix-prefix-mbonacci}.

As a conclusion to both cases, the longest palindrome occurring before is $c_{n+1} = Q(n+1)$ as required.
%
%
\end{proof}

\section{Open Problems}

\begin{problem}
Let $A=\{0,1, \ldots, m-1\}$ be a finite alphabet of size $m$, $m\ge 1$.
Define the family $\mathcal{F}_m$ of infinite words $\boldsymbol{w}$ over $A$ such that $z(\boldsymbol{w})=pz(\boldsymbol{w})$.
For instance, when $m=2$, observe that $\fib \in \mathcal{F}_2$ (Proposition~\ref{pro:fact-Fici-Fibonacci} and Theorem~\ref{thm:pal-z-fact-fibonacci}) and 
$$
010^21^20^31^3 \cdots \in \mathcal{F}_2,
$$
but the Thue--Morse word 
$$
\boldsymbol{t} = 0110100110010110 \cdots
$$
which is the fixed point of the morphism $\tau : 0 \to 01, 1 \to 10$ does not belong to $\mathcal{F}_2$.
Give a characterization of this family. 
Also give a characterization of the set of automatic words among this family.
\end{problem}

The $m$-bonacci words belong to the family of episturmian words.
When studying episturmian words, it is standard to introduce a particular sequence $(h_n)_{n\ge 0}$ of finite words related to their directive word and palindromic prefixes. 
Justin and Pirillo \cite{Justin-Pirillo02} showed that the sequence $(u_n)_{n\ge 1}$ of palindromic prefixes of a standard episturmian word $\boldsymbol{s}$ verifies
$$
u_{n+2} 
= h_{n} h_{n-1} \cdots h_1 h_0 
= h_0^R h_1^R \cdots h_{n-1}^R h_{n}^R \quad \forall \, n\ge 0.
$$
When it comes to $m$-bonacci words, the sequence $(h_n)_{n\ge 0}$ coincides with the one from Definition~\ref{def:mbonacci}.
From Lemma~\ref{lem:z-n-mbonacci}, we can also show that for all $n\ge 0$, we have 
$$
u_{n+2} = z_0 z_1 \cdots z_{n-2} z_{n-1} z_n  z_{n-1} z_{n-2}\cdots z_1 z_0, 
$$
which in particular gives another way of showing that the words $(z_n)_{n\ge -1}$ are all palindromes. 
We observe that the sequences $(h_n)_{n\ge 0}$ and $(z_n)_{n\ge -1}$ are intimately bonded, so a natural question is the following open problem.

\begin{problem}
Find the palindromic $z$- and $c$-factorizations of other infinite words such as the Thue--Morse word, or more specifically episturmian words, billiard words, or rich words.
\end{problem}

As we observed in Section~\ref{sec:p-singular},
Lemma~\ref{lem:z-n-mbonacci} gives another definition of the
p-singular words, which may possibly be the more useful one when trying to
extend the results of this paper to episturmian words (see the
ordinary $z$- and $c$-factorizations of episturmian words given by
Ghareghani, Mohammad-Noori, and Sharifani \cite{GMS11}).

\section*{Acknowledgements}

We thank the anonymous referee for their useful and thorough proofreading of a first draft of this paper. Their comments helped us significantly.

Morteza Mohammad-noori is supported in part by a grant from IPM No. 96050113.
Narad Rampersad is supported by NSERC Discovery Grant 418646-2012.
Manon Stipulanti is supported by FRIA Grant 1.E030.16.

\end{document}